\documentclass[11pt,english,british]{extarticle}
\usepackage{mathptmx}
\usepackage[T1]{fontenc}
\usepackage[latin9]{inputenc}
\usepackage{color}
\usepackage{cellspace}
\usepackage{babel}
\usepackage{float}
\usepackage{amsmath}
\usepackage{amsthm}
\usepackage{amssymb}
\usepackage{makecell}
\usepackage{graphicx}
\usepackage{subcaption}
\usepackage[a4paper]{geometry}
\usepackage[numbers]{natbib}
\usepackage[pdfusetitle,
 bookmarks=true,bookmarksnumbered=true,bookmarksopen=false,
 breaklinks=false,pdfborder={0 0 1},backref=false,colorlinks=true]
 {hyperref}
\hypersetup{
 pdfborderstyle=,linkcolor=black,citecolor=blue,urlcolor=black,filecolor=blue,pdfpagelayout=OneColumn,pdfnewwindow=true,pdfstartview=XYZ,plainpages=false}

\makeatletter
\theoremstyle{plain}
\newtheorem{thm}{\protect\theoremname}
\theoremstyle{plain}
\newtheorem{lem}[thm]{\protect\lemmaname}

\@ifundefined{date}{}{\date{}}
\usepackage{babel}
\usepackage{caption}
\usepackage[nottoc]{tocbibind}
\allowdisplaybreaks[4]
\usepackage{dsfont}
\DeclareMathAlphabet{\mathcal}{OMS}{cmsy}{m}{n}

\providecommand{\lemmaname}{Lemma}

\providecommand{\theoremname}{Theorem}

\makeatother

\addto\captionsbritish{\renewcommand{\lemmaname}{Lemma}}
\addto\captionsbritish{\renewcommand{\theoremname}{Theorem}}
\addto\captionsenglish{\renewcommand{\lemmaname}{Lemma}}
\addto\captionsenglish{\renewcommand{\theoremname}{Theorem}}
\providecommand{\lemmaname}{Lemma}
\providecommand{\theoremname}{Theorem}

\begin{document}
\title{Tracking Brownian fluid particles\\
 in large-eddy simulations}
\author{By Zihao Guo\thanks{Zhongtai Securities Institute for Financial Studies, Shandong University,
Jinan, China, 250100, Email: \protect\protect\href{mailto:gzhsdu@mail.sdu.edu.cn}{gzhsdu@mail.sdu.edu.cn}} \ and Zhongmin Qian\thanks{Mathematical Institute, University of Oxford, Oxford OX2 6GG, UK.,
and OSCAR, Suzhou, China. Email: \protect\protect\protect\protect\protect\protect\href{mailto:qianz@maths.ox.ac.uk}{qianz@maths.ox.ac.uk}}\ }
\maketitle
\begin{abstract}
In this paper, we propose an approach for simulating wall-bounded incompressible turbulent flows by integrating the technology of random vortex method with the core principles of large-eddy simulations (LES). In particular, we employ the filtering function, interpreted as a spatial averaging operator, together with the integral representation theorem for parabolic equations, to construct a closed numerical scheme suitable for computing solutions to the Navier-Stokes equations. This framework numerically overcomes the difficulties associated with the non-locally integrable three-dimensional kernel inherent in the random vortex method, enabling efficient computation of flow fields via the Monte Carlo method. Several numerical experiments are presented for both laminar and turbulent flows in wall-bounded domains, to thereby reveal the underlying flow mechanisms near the wall boundary. The experimental results and systematic comparisons with alternative numerical approaches consistently demonstrate that the proposed method is numerically stable, possesses low theoretical complexity, and achieves acceptable computational efficiency.

\medskip{}


\medskip{}

\end{abstract}

\section{Introduction}

In the present paper, we propose a novel scheme for implementing 
large-eddy simulations (LES, cf. \citep{Smagorinsky1963,Lilly1967, Deardorff1974}, also cf. the monographs
\citep{Berseli-LES, LesieurMetaisComte2005}) for wall-bounded viscous and incompressible fluid flows, based on a stochastic integral representation
of solutions to linear parabolic equations (see Appendix \ref{representation}
below) to be established in the present paper.   
LES is a numerical method extensively used for simulating turbulent flows (cf.
\citep{LesieurMetaisComte2005,Volker2004,Kumar2025}). In this article, we
introduce a natural closure to LES models by using a stochastic formulation of viscous incompressible fluid flows in terms of Brownian fluid particles. 

An incompressible viscous flow past a flat plate is modeled with a time-dependent velocity
\[u(x,t)=(u^{1}(x,t),\cdots,u^{d}(x,t)), \]
in $D:=\mathbb{R}_{+}^{d}=\{x:x\in\mathbb{R}^{d},x_{d}>0\}$
(where $d=2$ or $3$), so its boundary where $x_{d}=0$
models a solid wall. The fluid density $\rho(x,t)$ and pressure
$p(x,t)$ are other fluid flow quantities in this model. Due to incompressibility of the fluid flow,  $\rho(x,t)$
is constant, thus, without losing generality, we may assume that
$\rho(x,t)=1$ for all $x$ and $t$. The dynamics of the fluid flow
is described by the Navier-Stokes equations 
\begin{equation}
\frac{\partial}{\partial t}u^{i}+(u\cdot\nabla)u^{i}=\nu\Delta u^{i}-\frac{\partial}{\partial x_{i}}p+F^{i},\label{wbf-NS1}
\end{equation},
for $i=1,\ldots,d$, and 
\begin{equation}
\nabla\cdot u=0,\label{wbf-NS2}
\end{equation}
in $D\times[0,\infty)$, subject to the non-slip condition that $u(x,t)=0$
for $x\in\partial D$ and $t>0$. Here $\nu>0$ is the kinematic viscosity,
and $F(x,t)=(F^{1}(x,t),\cdots,F^{d}(x,t))$ stands for an external
force applying on the fluid flow. With increasing computational power,
numerical methods like finite difference, finite element method, fast
Fourier transform and etc. may be applied to calculating numerical
solutions of (\ref{wbf-NS1}, \ref{wbf-NS2}).

The Navier-Stokes equations may be formulated in various weak solution forms, which allows implementing the finite difference method of solving numerically their solutions. To this end, the mild solution approach seems useful. The simplest mild solution formulation is achieved by using the heat kernel. Let  $h(x,t,y)$ be the heat kernel of the heat operator $\nu\Delta-{\partial \over\partial t}$. For any given $t>0$, considering $f(s)=\int u(y,t-s)h(x,s,y)\textrm{d}y$. Then $f(0)=u(x,t)$, $f(t)=u_0(x)$, and
\begin{align}
    f(t)-f(0)&=\int_0^t f'(s)\textrm{d}s\notag\\
    &=\int_0^t\int_D u(y,t-s){\partial\over\partial s}h(x,s,y)\textrm{d}y\textrm{d}s-\int_0^t\int_D {\partial\over\partial s}u(y,t-s) h(x,s,y)\textrm{d}y \textrm{d}s\notag \\
    &=\int_0^t\int_D \left(\nu\Delta-{\partial\over\partial s}\right)u(y,t-s) h(x,s,y)\textrm{d}y\textrm{d}s\notag\\
    &=\int_0^t \int_D \left( u\cdot\nabla u +\nabla p - F\right)(y,t-s) h(x,s,y)\textrm{d}y\textrm{d}s,
    \end{align} 
which gives rise to the following implicit representation of the velocity
\begin{align}
u(x,t)=u_0(x)-\int_0^t \int_D \left( u\cdot\nabla u +\nabla p - F\right)(y,s) h(x,t-s,y)\textrm{d}y\textrm{d}s,
 \end{align} 
for every $t>0$ and $x\in D$. The gradient $\nabla u$ appearing in the integral on the right-hand side can be removed by using integration by parts due to the incompessibility of $u$, and 
\begin{align}
    &u^i(x,t) - u^i_0(x)=  \nonumber \\
    & \int_0^t \int_D\left(u^i(y,s) u(y,s)\cdot \nabla_y \ln{h}(x,t-s,y)-{\partial\over\partial y_i} p(y,s) - F^i(y,s)\right) h(x,t-s,y)\textrm{d}y\textrm{d}s,
    \end{align}
for $i=1,2,3$.  This mild solution setting may be used to implement finite difference 
method and calculate numerical solutions. However, this method requires still, though not explicitly involving the gradient of $u$ in the formulation,  the iteration of $\nabla u$ in order to calculate the pressure gradient $\nabla p$. This means that we need to use finite difference method to calculate $\nabla u$, which brings the risk of numerical explosion. There is another technical obstacle in this mild solution formulation, that is, the time of both singular integral kernel $\nabla \ln{h}$ and $h(x,t-s,y)$ are reversed, which implies that we need to save and use the velocity of every time during the iterative computations, causing a very long computing time. The mild solution formulation is widely used in the study of stochastic partial differential equations, it seems not suitable for directly numerical simulation of turbulent flows, further detailed explanation and comparison on this aspect shall be addressed in subsection \ref{comparison}.

To overcome this difficulty of the previous mild solution formulation one may try to follow the main idea in this formulation by absorbing the quadratic term $u\cdot \nabla u$ into the Laplacian term and make use of the fundamental solution associated with the forward parabolic operator $\nu \Delta + u\cdot \nabla - {\partial \over \partial t}$ in place of the Gaussian heat kernel $h(x,t,y)$, which is the transition probability of the diffusion with velocity $u(x,t)$. This approach can be traced back to the original work by Chorin \citep{Chorin1973} which gives the life of the random vortex method, which leads to the increasing attention paid to vortex \citep{Chang1991, Zhu2019,Zaboli2024,Cherepanov2025}. 

Let us review the main ideas in this approach for more details, cf. \citep{Chorin1973,CottetKoumoutsakos2000} and \citep{Majda and Bertozzi 2002}. In the random vortex method, the main fluid
dynamic variable is taken to be the vorticity, $\omega=\nabla\wedge u$, whose
dynamics is described by the vorticity transport equation 
\begin{equation}
\frac{\partial}{\partial t}\omega^{i}+(u\cdot\nabla)\omega^{i}=\nu\Delta\omega^{i}+(\omega\cdot\nabla)u^{i}+(\nabla\wedge F)^{i},\label{VORT-1}
\end{equation}
where $\omega^{i}=\varepsilon^{ijk}\frac{\partial}{\partial x_{j}}u^{k}$
for $i=1,2,3$. The flow velocity $u$ is recovered by solving the
Poisson equation that $\Delta u=-\nabla\wedge\omega$. 

It is appropriate to set up a convention about two-dimensional (2D)
fluid flows. By a 2D flow we mean a three-dimensional (3D) flow with velocity $u=(u^{1},u^{2},u^{3})$
where $u^{3}=0$, and $u^{1}$, $u^{2}$ depend only on coordinates $x_{1}$,
$x_{2}$. Hence, for a 2D flow, $\omega^{1}=0$, $\omega^{2}=0$
identically, $\omega^{3}=\frac{\partial}{\partial x_{1}}u^{2}-\frac{\partial}{\partial x_{2}}u^{1}$,
and $\omega$ is identified with a scalar function -- its third
component $\omega^{3}$, unless otherwise indicated. In particular
the vorticity stretching term $(\omega\cdot\nabla)u$ vanishes identically
for a 2D flow, and the vorticity equation becomes
\begin{equation}
\frac{\partial}{\partial t}\omega+(u\cdot\nabla)\omega=\nu\Delta\omega+\nabla\wedge F,
\end{equation}
where $\nabla\wedge F=\frac{\partial}{\partial x_{1}}F^{2}-\frac{\partial}{\partial x_{2}}F^{1}$
in this case. Hence, the vorticity $\omega$ can be represented in
terms of the fundamental solution $p(s,\xi,t,y)$ of the parabolic
operator $L_{-u}-\frac{\partial}{\partial t}$ (where we have introduced
a notation that $L_{-u}=\nu\Delta-u\cdot\nabla$), and the initial
vorticity $\omega_{0}$. Indeed
\begin{equation}
\omega(x,t)=\int_{\mathbb{R}^{2}}p(0,\xi,t,y)\omega_{0}(\xi)\textrm{d}\xi+\int_{0}^{t}\int_{\mathbb{R}^{2}}p(s,\xi,t,y)\nabla\wedge F(\xi,s)\textrm{d}\xi\textrm{d}s.\label{2D-v}
\end{equation}
Since $\nabla\cdot u=0$, $p(s,\xi,t,y)$ is the transition probability
density of the Brownian fluid particles $X$ with velocity $u$. By
combining with the Biot-Savart law, the previous integral representation
yields an integral representation for the velocity 
\begin{equation}
u(x,t)=\int_{\mathbb{R}^{2}}\mathbb{E}\left[K(x,X_{t}^{0,\xi})\omega_{0}(\xi)\right]\textrm{d}\xi+\int_{0}^{t}\int_{\mathbb{R}^{2}}\mathbb{E}\left[K(x,X_{t}^{s,\xi})\nabla\wedge F(\xi,s)\right]\textrm{d}\xi\textrm{d}s,\label{2D-vv}
\end{equation}
where $X^{s,\xi}$ denotes the Taylor diffusion initiated from location
$\xi$ at instance $s\geq0$, that is, the trajectories of Brownian
fluid particles with velocity $u$ determined by It\^o's stochastic
differential equations 
\begin{equation}
\textrm{d}X_{t}^{s,\xi}=u(X_{t}^{s,\xi},t)\textrm{d}t+\sqrt{2\nu}\textrm{d}B_{t},\quad X_{s}^{s,\xi}=\xi,\label{Taylor-01}
\end{equation}
where $B_{t}(t\geq s)$ is Brownian motion on a probability space.
$X$ is a diffusion process with infinitesimal generator $L_{u}$
(which is the $L^{2}$-adjoint operator of $L_{-u}$ as $\nabla\cdot u=0$). 

Two equations (\ref{2D-vv}, \ref{Taylor-01}) consist of a closed
McKean-Valsov type stochastic differential equations, which in turn
may be used in the design of numerical schemes for calculating the
flow velocity. It is also feasible to derive a random vortex formulation for flows constrained on a finite region and for flows past a solid wall. Although the random vortex method is mathematically elegant, it is effective
only for certain types of flows and is sometimes computationally expensive.

It is possible to formulate the random vortex scheme for three-dimensional viscous flows, which, however, involves an integral kernel which is not locally integrable, leading to difficulties in numerical computation of 3-dimensional flows. This poses a significant obstacle to the application of the random vortex method in the simulation of three-dimensional flows. In comparison, the random vortex method is more suitable for handling two-dimensional cases because the equations are simpler, and the integral kernel can be processed using a smoothing technique. See more details and related work in \citep{Guo2025} and \citep{Qian2022}. The significance, at least formally, the quadratic term $u\cdot \nabla u$ is hidden in the previous integral representation of the velocity in terms of the Brownian particle diffusion, which appears linear in $u$.  

If a viscous fluid flow moves along a solid wall with a high speed, then substantial vorticity may be created at the boundary which in
turn generates turbulent motion in a thin layer next to the boundary, and may lead to flow separation from the wall. The small-scale energy dissipation in turbulent motion may be described by Kolmogorov's theory, known as the K41 theory, while the global motion of boundary turbulence may be calculated by means of LES, or other numerical schemes. For simulations of turbulent flows, methods such as finite difference and finite element require extremely fine mesh resolution, which places high demands on computational power, therefore, schemes for calculating solutions of turbulent flows from their dynamical equations at separate space scales shall be advocated. LES method (cf. \citep{Piomelli2002,Smagorinsky1963}), however, uses a filtering function to separate large eddies from small eddies: the large eddies are computed numerically, while the small eddies are modeled using standardized models. On a small scale, the universality principle of turbulence implies that the statistical properties of turbulence, such as energy dissipation and the structure of velocity gradients, become independent of large-scale features of the flow and are governed by universal scaling laws \citep{Smagorinsky1963,Germano1991,Kang2023,Kim2024}. The integration of LES and machine learning discovers superior closure models and algorithms in a data-driven manner, significantly enhancing computational efficiency while ensuring accuracy \citep{Zhou2025,Zhang2025}. LES approach greatly reduces computational cost and has become the primary method for turbulence simulation today. Compared to the random vortex method, LES suffers from significant grid dependency, inherent numerical dissipation that smears small-scale vortices, and greater complexity in handling moving boundaries. This motivates us to develop a new random vortex method that focuses more on relatively large-scale flows. 

To address the challenges faced by the random vortex method in three-dimensional simulations and to incorporate the concept of large-eddy modeling, we aim to propose a random calculation method for implementing LES. The contribution of this paper lies in combining LES with the stochastic integral representation theorem to propose an explicit, forward-in-time computational method for the Navier-Stokes equations. This random LES approach can be adopted to simulate three-dimensional wall-bounded turbulence using the Monte Carlo method. According to our experimental results, our method is tolerant of mesh resolution, does not require high computational power, remains numerically stable without blow-up, and has acceptable computation time. Moreover, the method supports repeatability, allowing the simulation of fluid behavior under potential distributions. Our experimental results reveal some of the flow mechanisms at the wall boundary. Compared with the stochastic vortex method, our approach overcomes the challenge of handling three-dimensional fluid systems. Compared with traditional LES, it closes the Navier-Stokes equations without further modelling the stress tensor and imposes less stringent demands near the wall.

The remainder of the paper is organized as follows. In Section 2, the theoretical foundation of the random LES, proposed in this work, will be established including the details of the formulation of incompressible viscous fluid flows in terms of Brownian fluid particles. In Section 3, the random LES flows past a wall is established. The random LES are formulated in terms of stochastic differential equations involving the distribution of Brownian fluid particles. Various numerical
experiments based on random LES, comparisons, and analysis are reported in Section 4. Finally, in Appendix, we include the necessary results on the functional integral representation for solutions of linear parabolic systems, used in the formulation of the random LES. 

\section{Theoretical aspect of random LES}

Let us first reformulate the fluid dynamics in terms of Brownian fluid
particles, following G.Taylor \citep{Taylor1921}.
Although it has not been established rigorously yet, we will nevertheless
assume that the velocity $u(x,t)$ is bounded and smooth in $D\times[0,T)$
(for some $T>0$) and at least $C^{1}$ up to the boundary $\partial D$.
As a consequence, $p(x,t)$ shares the same regularity as that of $u(x,t)$.
The dynamics of the flow is determined by the Navier-Stokes equations
(\ref{wbf-NS1}, \ref{wbf-NS2}) together with the non-slip boundary
condition. We introduce the theoretical framework of random LES.

\subsection{Computation of the pressure gradient}

\label{cal_p}

Observe that there are two non-linear terms appearing in the Navier-Stokes
equations (\ref{wbf-NS1}), the non-linear convection term $(u\cdot\nabla)u$
and the pressure gradient $\nabla p$. In all numerical methods, one
has to deal with the pressure gradient $\nabla p$. In Direct Numerical
Simulation (DNS), one can select an appropriate basis, and project
the Navier-Stokes equations to the space of divergence-free vector
fields for eliminating the pressure gradient $\nabla p$. However,
this method works well only for special class of regions. In random
vortex method, one instead takes the vorticity $\omega=\nabla\wedge u$
to be the main fluid dynamical variable, and projects the Navier-Stokes
equation to the differential of $u$, which is the vorticity transport
equation. One can represent the vorticity $\omega$ in terms of Brownian
fluid particles $X$. This also allows for the application of finite
difference or finite element methods to both the vorticity transport
equation and the Poisson equation.

The method we are going to propose, which combines the idea of Brownian
fluid particles with an idea from LES, is based on a functional integral
representation of solutions to linear parabolic equations in terms
of distributions of Brownian fluid particles, cf. Appendix \ref{representation}.
Unlike the random vortex approach, although still utilizing the distribution
of the Brownian fluid particles, we calculate the velocity directly
instead of dealing with the vorticity equation. Hence we need to resolve
the pressure gradient $\nabla p$ first. 

Let $D=\mathbb{R}_{+}^{d}$ (where $d=2$ or $3$) whose boundary
$\partial D$ is described by the equation that $x_{d}=0$, and is
identified with $\mathbb{R}^{d-1}$. The unit normal pointing outwards
to $\partial D$ is a constant vector $\boldsymbol{n}=(0,\cdots,0,-1)$,
and the normal derivative $\frac{\partial}{\partial\boldsymbol{n}}=\left.-\frac{\partial}{\partial x_{d}}\right|_{x_{d}=0}$. 

Since $\nabla\cdot u=0$, the pressure $p(x,t)$ at every instance
$t>0$ is recovered by solving the Poisson equation 
\begin{equation}
\Delta p=-\frac{\partial u^{j}}{\partial x_{i}}\frac{\partial u^{i}}{\partial x_{j}}+\nabla\cdot F.\label{Eq-pressure}
\end{equation}
Suppose the Navier-Stokes equations continue to be valid up to the
boundary, then the pressure $p$ at the boundary $\partial D$ is
subject to the Neumann boundary condition that 
\begin{equation}
\frac{\partial}{\partial\boldsymbol{n}}p=\nu\left(\frac{\partial}{\partial x_{1}}\frac{\partial u^{1}}{\partial x_{3}}+\frac{\partial}{\partial x_{2}}\frac{\partial u^{2}}{\partial x_{3}}\right)-F^{3}\quad\textrm{ on }\partial D,\label{NBV-1}
\end{equation}
where the right-hand side is evaluated along $\partial D$. Since $\frac{\partial}{\partial\boldsymbol{n}}=\left.-\frac{\partial}{\partial x_{3}}\right|_{x_{3}=0}$,
the boundary condition may be written as 
\begin{equation}
\left.\frac{\partial p}{\partial x_{3}}\right|_{x_{3}=0}=-\nu\left(\frac{\partial}{\partial x_{1}}\frac{\partial u^{1}}{\partial x_{3}}+\frac{\partial}{\partial x_{2}}\frac{\partial u^{2}}{\partial x_{3}}\right)+F^{3}\quad\textrm{ on }x_{3}=0.\label{p-boun}
\end{equation}
Similarly for 2D flows
\begin{equation}
\left.\frac{\partial p}{\partial x_{2}}\right|_{x_{2}=0}=-\nu\frac{\partial}{\partial x_{1}}\frac{\partial u^{1}}{\partial x_{2}}+F^{2}\quad\textrm{ on }x_{2}=0.\label{p-boun-1}
\end{equation}

We shall now recall the integral representation for solutions to the
boundary problem (\ref{p-boun}) of the Poisson equation (\ref{Eq-pressure}). According to the reflection principle, the Green function for $\mathbb{R}_{+}^{d}$,
subject to the Neumann boundary condition, denoted by $H_{+}^{d}(x,y)$,
is given by the following formula
\begin{equation}
H_{d}^{+}(x,y)=\Gamma_{d}(x-y)+\Gamma_{d}(x-\bar{y})\quad\textrm{ for }x\neq y,\bar{y},\label{H-Positive-G1}
\end{equation}
where $\Gamma_{d}$ is the elementary solution to the Laplace equation
on $\mathbb{R}^{d}$:
\begin{equation}
\Gamma_{2}(x)=\frac{1}{2\pi}\ln|x|\quad\textrm{ and }\Gamma_{d}(x)=-\frac{1}{(d-2)s_{d-1}}\frac{1}{|x|^{d-2}}\textrm{ for }x\neq0,\label{Elem-01}
\end{equation}
for $d\geq3$, where $s_{d-1}$ is the area of $d-1$ dimensional
sphere, $s_{1}=2\pi$ and $s_{2}=4\pi$.

The Green formula yields the following representation (for simplicity
the time variable $t$ in $p(x,t)$ is suppressed)
\begin{align}
p(x) & =\int_{\mathbb{R}_{+}^{d}}H_{d}^{+}(x,y)\Delta p(y)\textrm{d}y-\int_{\left\{ y_{d}=0\right\} }H_{d}^{+}(x,y)\frac{\partial p}{\partial\boldsymbol{n}}(y)\textrm{d}y_{1}\cdots\textrm{d}y_{d-1}\nonumber \\
 & =\int_{\mathbb{R}_{+}^{d}}H_{d}^{+}(x,y)\Delta p(y)\textrm{d}y\nonumber \\
 & \;+\int_{\mathbb{R}^{d-1}}H_{d}^{+}(x,(y_{1},\cdots,y_{d-1},0))\left.\frac{\partial p}{\partial y_{d}}\right|_{y_{d}=0}(y_{1},\cdots,y_{d-1},0)\textrm{d}y_{1}\cdots\textrm{d}y_{d-1},\label{p-rep-01}
\end{align}
for $x\in\mathbb{R}_{+}^{d}$. Noticing that we are only interested
in the gradient of pressure $\nabla p$, which can be calculated therefore
by the following formula:
\begin{align}
\nabla p(x) & =\int_{\mathbb{R}_{+}^{d}}K_{d}^{+}(x,y)\Delta p(y)\textrm{d}y\nonumber \\
 & \;+\int_{\mathbb{R}^{d-1}}K_{d}^{+}(x,(y_{1},\cdots,y_{d-1},0))\left.\frac{\partial p}{\partial y_{d}}\right|_{y_{d}=0}(y_{1},\cdots,y_{d-1},0)\textrm{d}y_{1}\cdots\textrm{d}y_{d-1},\label{p-gradP-01}
\end{align}
where $K_{d}^{+}(x,y)=\nabla_{x}H_{d}^{+}(x,y)$, the Biot-Savart
kernel for the half space $\mathbb{R}_{+}^{d}$. 

If $d=2$, then 
\begin{equation}
K_{2}^{+}(x,y)=\frac{1}{2\pi}\left(\frac{x-y}{|x-y|^{2}}+\frac{x-\bar{y}}{|x-\bar{y}|^{2}}\right)\quad\textrm{ for }x\neq y,\bar{y},\label{K2+-01}
\end{equation}
and 
\begin{equation}
K_{2}^{+}((x_{1},x_{2}),(y_{1},0))=\frac{1}{\pi}\frac{(x_{1}-y_{1},x_{2})}{|x_{1}-y_{1}|^{2}+|x_{2}|^{2}}\quad\textrm{ for }x_{2}\neq0.\label{K2+-02}
\end{equation}

If $d=3$, then
\begin{equation}
K_{3}^{+}(x,y)=\frac{1}{4\pi}\left(\frac{x-y}{|x-y|^{3}}+\frac{x-\bar{y}}{|x-\bar{y}|^{3}}\right)\quad\textrm{ for }x\neq y,\bar{y},\label{K3+-01}
\end{equation}
and
\begin{equation}
K_{3}^{+}((x_{1},x_{2},x_{3}),(y_{1},y_{2},0))=\frac{1}{2\pi}\frac{\left(x_{1}-y_{1},x_{2}-y_{2},x_{3}\right)}{\left(|x_{1}-y_{1}|^{2}+|x_{2}-y_{2}|^{2}+|x_{3}|^{2}\right)^{\frac{3}{2}}},\label{K3+-02}
\end{equation}
for $x_{3}\neq0$. 
\begin{lem}
\label{lem1-2Dflows}(2D flows). For a 2D incompressible viscous flow
the gradient of the pressure at any instance has the following singular
integral representation
\begin{align}
\nabla p(x) & =\int_{\mathbb{R}_{+}^{2}}K_{2}^{+}(x,y)\Delta p(y)\textrm{d}y+\int_{\mathbb{R}}K_{2}^{+}(x,(y_{1},0))F^{2}(y_{1},0)\textrm{d}y_{1}\nonumber \\
 & \;+\nu\int_{\mathbb{R}}\frac{\partial K_{2}^{+}}{\partial y_{1}}(x,(y_{1},0))\frac{\partial u^{1}}{\partial y_{2}}(y_{1},0)\textrm{d}y_{1},\label{grad-p-01}
\end{align}
for $x\in\mathbb{R}_{+}^{2}$, where the last singular integral kernel
is given by
\begin{equation}
\frac{\partial K_{2}^{+}}{\partial y_{1}}((x_{1},x_{2}),(y_{1},0))=\frac{1}{\pi}\left(\frac{|x_{1}-y_{1}|^{2}-|x_{2}|^{2}}{\left(|x_{1}-y_{1}|^{2}+|x_{2}|^{2}\right)^{2}},\frac{2x_{2}(x_{1}-y_{1})}{\left(|x_{1}-y_{1}|^{2}+|x_{2}|^{2}\right)^{2}}\right)\label{grad-K-3rd-01}
\end{equation}
for $x_{2}>0$.
\end{lem}

\begin{proof}
By using the integral representation
\begin{align*}
\nabla p(x) & =\int_{\mathbb{R}_{+}^{2}}K_{2}^{+}(x,y)\Delta p(y)\textrm{d}y\\
 & \;+\int_{\mathbb{R}^{1}}K_{2}^{+}(x,(y_{1},0))\left(-\nu\frac{\partial}{\partial y_{1}}\frac{\partial u^{1}}{\partial y_{2}}+F^{2}\right)(y_{1},0)\textrm{d}y_{1}.
\end{align*}
For the integral involving the iterated derivative of $u$, we may
perform integration by parts we obtain that
\begin{align*}
\nabla p(x) & =\int_{\mathbb{R}_{+}^{2}}K_{2}^{+}(x,y)\Delta p(y)\textrm{d}y+\int_{\mathbb{R}^{1}}K_{2}^{+}(x,(y_{1},0))F^{2}(y_{1},0)\textrm{d}y_{1}\\
 & \;+\nu\int_{\mathbb{R}^{1}}\frac{\partial K_{2}^{+}}{\partial y_{1}}(x,(y_{1},0))\frac{\partial u^{1}}{\partial y_{2}}(y_{1},0)\textrm{d}y_{1},
\end{align*}
where the kernel for the third singular integral can be worked out
as
\begin{align*}
\frac{\partial K_{2}^{+}}{\partial y_{1}}((x_{1},x_{2}),(y_{1},0)) & =\frac{1}{\pi}\frac{(-1,0)}{|x_{1}-y_{1}|^{2}+|x_{2}|^{2}}+\frac{2}{\pi}\frac{\left(|x_{1}-y_{1}|^{2},x_{2}(x_{1}-y_{1})\right)}{\left(|x_{1}-y_{1}|^{2}+|x_{2}|^{2}\right)^{2}}\\
 & =\frac{1}{\pi}\left(\frac{|x_{1}-y_{1}|^{2}-|x_{2}|^{2}}{\left(|x_{1}-y_{1}|^{2}+|x_{2}|^{2}\right)^{2}},\frac{2x_{2}(x_{1}-y_{1})}{\left(|x_{1}-y_{1}|^{2}+|x_{2}|^{2}\right)^{2}}\right).
\end{align*}
\end{proof}
By utilising the same argument, we may obtain a similar representation
theorem stated as the following.
\begin{lem}
\label{lem2-3Dflows}(3D flows) For 3D flows in $\mathbb{R}_{+}^{3}$
the gradient of the pressure is represented as
\begin{align}
\nabla p(x) & =\int_{\mathbb{R}_{+}^{3}}K_{3}^{+}(x,y)\Delta p(y)\textrm{d}y+\int_{\mathbb{R}^{2}}K_{3}^{+}(x,(y_{1},y_{2},0))F^{3}(y_{1},y_{2},0)\textrm{d}y_{1}\textrm{d}y_{2}\nonumber \\
 & +\nu\int_{\mathbb{R}^{2}}\frac{\partial K_{3}^{+}}{\partial y_{1}}(x,(y_{1},y_{2},0))\frac{\partial u^{1}}{\partial y_{3}}(y_{1},y_{2},0)\textrm{d}y_{1}\textrm{d}y_{2}\nonumber \\
 & +\nu\int_{\mathbb{R}^{2}}\frac{\partial K_{3}^{+}}{\partial y_{2}}(x,(y_{1},y_{2},0))\frac{\partial u^{2}}{\partial y_{3}}(y_{1},y_{2},0)\textrm{d}y_{1}\textrm{d}y_{2},\label{grad-p3D-01}
\end{align}
where 
\begin{equation}
\frac{\partial K_{3}^{+}}{\partial y_{1}}(x,y)=\frac{\left(2|x_{1}-y_{1}|^{2}-|x_{2}-y_{2}|^{2}-|x_{3}|^{2},3(x_{2}-y_{2})(x_{1}-y_{1}),3x_{3}(x_{1}-y_{1})\right)}{2\pi\left(|x_{1}-y_{1}|^{2}+|x_{2}-y_{2}|^{2}+|x_{3}|^{2}\right)^{\frac{5}{2}}},\label{3rd-ke-1}
\end{equation}
and
\begin{equation}
\frac{\partial K_{3}^{+}}{\partial y_{2}}(x,y)=\frac{\left(3(x_{1}-y_{1})(x_{2}-y_{2}),2|x_{2}-y_{2}|^{2}-|x_{1}-y_{1}|^{2}-|x_{3}|^{2},3x_{3}(x_{2}-y_{2})\right)}{2\pi\left(|x_{1}-y_{1}|^{2}+|x_{2}-y_{2}|^{2}+|x_{3}|^{2}\right)^{\frac{5}{2}}}.\label{3rd-ke-2}
\end{equation}
for $x=(x_{1},x_{2},x_{3})$ with $x_{3}>0$ and $y=(y_{1},y_{2},0)$.
\end{lem}

\begin{proof}
Indeed, by the Green formula, and using the boundary condition for
$p$, we obtain
\begin{align*}
\nabla p(x) & =\int_{\mathbb{R}_{+}^{3}}K_{3}^{+}(x,y)\Delta p(y)\textrm{d}y+\int_{\mathbb{R}^{2}}K_{3}^{+}(x,(y_{1},y_{2},0))\left.\frac{\partial p}{\partial y_{3}}\right|_{y_{3}=0}(y_{1},y_{2},0)\textrm{d}y_{1}\textrm{d}y_{2}\\
 & \;=\int_{\mathbb{R}_{+}^{3}}K_{3}^{+}(x,y)\Delta p(y)\textrm{d}y+\int_{\mathbb{R}^{2}}K_{3}^{+}(x,(y_{1},y_{2},0))F^{3}(y_{1},y_{2},0)\textrm{d}y_{1}\textrm{d}y_{2}\\
 & \;+\int_{\mathbb{R}^{2}}K_{3}^{+}(x,(y_{1},y_{2},0))\left.-\nu\left(\frac{\partial}{\partial y_{1}}\frac{\partial u^{1}}{\partial y_{3}}+\frac{\partial}{\partial y_{2}}\frac{\partial u^{2}}{\partial y_{3}}\right)\right|_{y_{3}=0}(y_{1},y_{2},0)\textrm{d}y_{1}\textrm{d}y_{2}.
\end{align*}
For the last integral we shall perform integration by parts, we deduce
that
\begin{align*}
\nabla p(x) & =\int_{\mathbb{R}_{+}^{3}}K_{3}^{+}(x,y)\Delta p(y)\textrm{d}y+\int_{\mathbb{R}^{2}}K_{3}^{+}(x,(y_{1},y_{2},0))F^{3}(y_{1},y_{2},0)\textrm{d}y_{1}\textrm{d}y_{2}\\
 & +\nu\int_{\mathbb{R}^{2}}\frac{\partial K_{3}^{+}}{\partial y_{1}}(x,(y_{1},y_{2},0))\frac{\partial u^{1}}{\partial y_{3}}(y_{1},y_{2},0)\textrm{d}y_{1}\textrm{d}y_{2}\\
 & +\nu\int_{\mathbb{R}^{2}}\frac{\partial K_{3}^{+}}{\partial y_{2}}(x,(y_{1},y_{2},0))\frac{\partial u^{2}}{\partial y_{3}}(y_{1},y_{2},0)\textrm{d}y_{1}\textrm{d}y_{2}.
\end{align*}
Here again those partial derivatives of $K_{3}^{+}$ can be worked
out explicitly.
\end{proof}
Finally let us mention the integral representation for viscous incompressible
flows without constraint. 
\begin{lem}
\label{lem3-whole}Suppose $p(x,t)$ is the pressure of an incompressible
fluid flow in $\mathbb{R}^{d}$ (where $d=2$ or $3$) with viscosity
constant $\nu>0$. Then
\begin{equation}
\nabla p(x)=\int_{\mathbb{R}^{d}}K_{d}(x,y)\Delta p(y)\textrm{d}y\quad\textrm{ for }x\in\mathbb{R}^{d},\label{grad-w-p}
\end{equation}
where $K_{d}(x,y)=\nabla_{x}\Gamma_{d}(x-y)$ is the Biot-Salve singular
integral kernel, given by explicitly
\begin{equation}
K_{d}(x,y)=\frac{1}{s_{d-1}}\frac{x-y}{|x-y|^{d}}\quad\textrm{ for }x\neq y,\label{BS-ke-01}
\end{equation}
and $s_{d-1}$ is the area of the $d-1$ dimensional unit sphere,
$s_{1}=2\pi$ and $s_{2}=4\pi$.
\end{lem}

\subsection{Representing the velocity in terms of fluid particles}

After obtaining the pressure gradient $\nabla p$ for any fixed $t$
we may update the velocity field $u(x,t)=(u^{1}(x,t),\cdots,u^{d}(x,t))$
by considering $\nabla p$ as an additional force term applying to
the flow. The Navier-Stokes equations may be written as a parabolic
equation 
\begin{equation}
\left(L_{-u}-\frac{\partial}{\partial t}\right)u+g=0\quad\textrm{ in }D,\label{par-01}
\end{equation}
subject to the non-slip condition that $u(x,t)=0$ for $x\in\partial D$
and the initial condition that $u(x,0)=u_{0}(x)$, where $L_{-u}=\nu\Delta-u\cdot\nabla$
and $g=-\nabla p+F$. Hence the velocity $u(x,t)$ may be represented
in terms of the distribution of the Brownian fluid particles. Recall
that the Brownian fluid particles with velocity $u(x,t)$ are governed
by It\^o's stochastic differential equation 
\begin{equation}
\textrm{d}X_{t}=u(X_{t},t)\textrm{d}t+\sqrt{2v}\textrm{d}B_{t},\quad X_{0}=\eta,\label{x-eta}
\end{equation}
where $B$ is $d$-dimensional Brownian motion on a probability space
$(\varOmega,\mathcal{F},\mathbb{P})$, and $\eta\in\mathbb{R}^{d}$
is the initial position of Brownian fluid particles. 

The particles defined by (\ref{x-eta}) initiated from $\eta$ is
denoted by $X^{\eta}$. Brownian particles defined by (\ref{x-eta})
form a family of diffusion with the infinitesimal generator $L_{u}=\nu\Delta+u\cdot\nabla$.
Let $p_{u}(s,x;t,y)$ be the transition probability density function
of Brownian fluid particles, that is, $p_{u}(s,x;t,y)=\mathbb{P}\left[\left.X_{t}=y\right|X_{s}=x\right]$
(for $t>s$) the probability of the fluid particles being at $y$
at time $t$ for those particles initiated from $x$ at time $s$. 
\begin{thm}
\label{thm-space}Suppose $u(x,t)$ is the velocity of an viscous
incompressible flow in $\mathbb{R}^{d}$. Then
\begin{equation}
u(\xi,t)=\int_{\mathbb{R}^{d}}p_{u}(0,\eta;t,\xi)u_{0}(\eta)\textrm{d}\eta+\int_{0}^{t}\int_{\mathbb{R}^{d}}\mathbb{E}\left[\left.g(X_{s}^{\eta},s)\right|X_{t}^{\eta}=\xi\right]p_{u}(0,\eta;t,\xi)\textrm{d}\eta\textrm{d}s,\label{rep-u-b01}
\end{equation}
for any $\xi\in\mathbb{R}^{d}$ and $t\geq0$, where $u_{0}=u(\cdot,0)$
is the initial velocity.
\end{thm}

Let us now consider a flow past a flat plate, i.e. an viscous incompressible
flow $u(x,t)=(u^{1}(x,t),\cdots,u^{d}(x,t))$ for $x\in\mathbb{R}_{+}^{d}$,
subject to the non-slip condition that $u(x,t)=0$ for $x=(x_{1},\cdots,x_{d})$
with $x_{d}=0$. 

The velocity $u(x,t)$ is extended to the whole space by reflection,
and therefore $u(\bar{x},t)=\overline{u(x,t)}$ for every $x\in\mathbb{R}^{d}$.
That is, $u^{i}(\bar{x},t)=u^{i}(x,t)$ ($i=1,\cdots,d-1$) and $u^{d}(\bar{x},t)=-u^{d}(x,t)$.
The extension $u(x,t)$ is then divergence-free in distribution on
$\mathbb{R}^{d}$ for each $t$, and therefore $L_{u}$ is the $L^{2}$-adjoint
of $L_{-u}$. Hence $p_{u}(s,x;t,y)$ is the fundamental solution
of the backward parabolic operator $L_{-u}+\frac{\partial}{\partial t}$.
According to the integral representation (cf. Theorem \ref{thm23-5}
below), we have the following representation theorem.
\begin{thm}
\label{thmflowpastw}For a viscous incompressible flow in $\mathbb{R}_{+}^{d}$,
the following integral representation holds:
\begin{align}
u(\xi,t) & =\int_{\mathbb{R}_{+}^{d}}\left(p_{u}(0,\eta;t,\xi)-p_{u}(0,\bar{\eta};t,\xi)\right)u_{0}(\eta)\textrm{d}\eta\nonumber \\
 & +\int_{0}^{t}\int_{\mathbb{R}_{+}^{d}}\mathbb{E}\left[\left.1_{\{t-s<\zeta(X^{\eta}) \}}g(X_{s}^{\eta},s)\right|X_{t}^{\eta}=\xi\right]p_{u}(0,\eta;t,\xi)\textrm{d}\eta\textrm{d}s,\label{u-int1}
\end{align}
for $\xi\in\mathbb{R}_{+}^{d}$ and $t>0$, and $u(\bar{\xi},t)=\overline{u(\xi,t)}$
for $\xi\notin\mathbb{R}_{+}^{d}$. Here $\zeta(w)=\inf\{s>0:w(s)\notin\mathbb{R}_{+}^{d}\}$
for every continuous path $w$ in $\mathbb{R}^{d}$. 
\end{thm}

See more proof processes and details in Appendix \ref{representation}. We therefore have formulated the Navier-Stokes equations in terms
of (\ref{u-int1}, \ref{x-eta}, \ref{grad-p-01}, \ref{grad-p3D-01}).
We shall use this formulation to implement large-eddy simulation
for viscous incompressible flows past a solid boundary.

\section{Random LES method}

We borrow a key idea from LES to carry out numerical experiments
based on the functional integral representation (\ref{u-int1}) for
viscous incompressible flows. The principal idea in LES is to facilitate
the calculation by ignoring the eddies in small length scales, which
are very expensive in computing, via filtering of the Navier--Stokes
equations. Such filtering, which can be seen as spatial averaging,
effectively removes the information of eddies in small-scale and saves
the computing cost, cf. \citep{Reynolds1895}. That is, LES method does not
aim to calculate the velocity $u(x,t)$, but only aims to compute the
local average $\tilde{u}(x,t)$ at every grid point. To implement
this approach, we choose a filter function $\chi$, which is a non-negative
function on $\mathbb{R}^{d}$, ideally with a compact support, such
that $\int_{\mathbb{R}^{d}}\chi(x)\textrm{d}x=1$. A possible choice
of the filter is the Gaussian kernel on $\mathbb{R}^{d}$ with variance
$\sigma^{2}>0$:
\begin{equation}
\chi(x)=\frac{1}{(2\pi\sigma^{2})^{\frac{d}{2}}}\exp\left(-\frac{|x|^{2}}{2\sigma^{2}}\right)\quad\textrm{ for }x\in\mathbb{R}^{d},
\end{equation}
where $\sigma^{2}>0$ should be small and is determined by the grid
size used in the numerical schemes. This filter function has its support on
$\mathbb{R}^{d}$, while when $\sigma^{2}>0$ is small enough, $\chi(x-y)$
is concentrated at about $x$. 

\subsection{Random LES for flows in $\mathbb{R}^{d}$}

For a flow in $\mathbb{R}^{d}$, the filtered velocity $\tilde{u}(x,t)$
is defined by
\begin{equation}
\tilde{u}(x,t)=\int_{\mathbb{R}^{d}}\chi(x-y)u(y,t)\mathrm{d}y\quad\textrm{ for }x\in\mathbb{R}^{d}.\label{w-av1}
\end{equation}
The filtered Navier-Stokes equations become 
\begin{equation}
\frac{\partial}{\partial t}\tilde{u}+\widetilde{(u\cdot\nabla)u}=\nu\Delta\tilde{u}-\nabla\tilde{p}+\tilde{F},
\end{equation}
and 
\begin{equation}
\nabla\cdot\tilde{u}=0.
\end{equation}
Because $\widetilde{(u\cdot\nabla)u}$ can not be simplified as $(\tilde{u}\cdot\nabla)\tilde{u}$,
so people define sub-grid-scale stress $\tau$ as 
\begin{equation}
\nabla\cdot\tau=(\tilde{u}\cdot\nabla)\tilde{u}-\widetilde{(u\cdot\nabla)u}.
\end{equation}
Therefore the filtered Navier-Stokes equation can be rewritten as
\begin{equation}
\partial_{t}\tilde{u}+(\tilde{u}\cdot\nabla)\tilde{u}=\nu\Delta\tilde{u}-\nabla\tilde{p}+\tilde{F}-\nabla\cdot\tau.
\end{equation}
The tensor $\tau$ measures the effect of small eddies on the flow
system which should be modelled. The influence of small-scale eddies
on the equation of motion is described by some other models (cf. \citep{Piomelli2002},
\citep{Smagorinsky1963} and \citep{Pitsch2006}), but it will not
be elaborated in our paper. LES effectively reduces the computational
complexity and retains the main physical properties. 
\begin{lem}
\label{lem-lse-r1}For a viscous incompressible in $\mathbb{R}^{d}$,
the following representation holds
\begin{equation}
\tilde{u}(x,t)=\int_{\mathbb{R}^{d}}\mathbb{E}\left[\chi(x-X_{t}^{\eta})u_{0}(\eta)\right]\textrm{d}\eta+\int_{0}^{t}\int_{\mathbb{R}^{d}}\mathbb{E}\left[\chi(x-X_{t}^{\eta})g(X_{s}^{\eta},s)\right]\textrm{d}\eta\textrm{d}s,\label{lse-w-1}
\end{equation}
for $x\in\mathbb{R}^{d}$. 
\end{lem}

\begin{proof}
Using the integral representation (\ref{rep-u-b01}) we obtain that
\begin{align*}
\tilde{u}(x,t) & =\int_{\mathbb{R}^{d}}\int_{\mathbb{R}^{d}}\chi(x-y)p_{u}(0,\eta;t,y)u_{0}(\eta)\textrm{d}y\textrm{d}\eta\\
 & \;+\int_{0}^{t}\int_{\mathbb{R}^{d}}\int_{\mathbb{R}^{d}}\chi(x-y)\mathbb{E}\left[\left.g(X_{s}^{\eta},s)\right|X_{t}^{\eta}=y\right]p_{u}(0,\eta;t,y)\textrm{d}y\textrm{d}\eta\textrm{d}s\\
 & =\int_{\mathbb{R}^{d}}\mathbb{E}\left[\chi(x-X_{t}^{\eta})u_{0}(\eta)\right]\textrm{d}\eta+\int_{0}^{t}\int_{\mathbb{R}^{d}}\mathbb{E}\left[\chi(x-X_{t}^{\eta})g(X_{s}^{\eta},s)\right]\textrm{d}\eta\textrm{d}s,
\end{align*}
which completes the proof.
\end{proof}
Let us now write down the numerical scheme of random LES.

\vskip0.3truecm

\begin{itemize}
    \item \textbf{\emph{Numerical schemes for random LES method}}
\end{itemize}

\vskip0.3truecm

Based on this representation, we are able to propose the following
random LES scheme for flows without space constrain as the following.
\begin{equation}
\textrm{d}Y_{t}^{\eta}=U(Y_{t}^{\eta},t)\textrm{d}t+\sqrt{2\nu}\textrm{d}B_{t},\quad Y_{0}^{\eta}=\eta,\label{w1}
\end{equation}
for every $\eta\in\mathbb{R}^{d}$, where $B=(B^{1},\cdots,B^{d})$
is a standard Brownian motion on a probability space $(\varOmega,\mathcal{F},\mathbb{P})$;
\begin{equation}
U(x,t)=\int_{\mathbb{R}^{d}}\mathbb{E}\left[\chi(x-Y_{t}^{\eta})u_{0}(\eta)\right]\textrm{d}\eta+\int_{0}^{t}\int_{\mathbb{R}^{d}}\mathbb{E}\left[\chi(x-Y_{t}^{\eta})G(Y_{s}^{\eta},s)\right]\textrm{d}\eta\textrm{d}s,\label{w2}
\end{equation}
for $x\in\mathbb{R}^{d}$, where $u_{0}$ is the initial velocity,
\begin{equation}
G(x,t)=-\nabla P(x,t)+F(x,t),\label{w3}
\end{equation}
for every $x\in\mathbb{R}^{d}$ and $t\geq0$,  $F=(F^{1},F^{2},F^{3})$
is an external force applying to the flow, and
\begin{equation}
\nabla P(x,t)=\int_{\mathbb{R}^{d}}K_{d}(x,y)\left.\left(\nabla\cdot F-\sum_{i,j=1}^{d}\frac{\partial U^{j}}{\partial y^{i}}\frac{\partial U^{i}}{\partial y^{j}}\right)\right|_{(y,t)}\textrm{d}y,\label{w4}
\end{equation}
for every $x\in\mathbb{R}^{d}$ and $t\geq0$, where $K_{d}(x,y)=\frac{1}{s_{d-1}}\frac{x-y}{|x-y|^{d}}$
for $x\neq y$, $s_{1}=2\pi$ and $s_{2}=4\pi$.

\subsection{Random LES for flows past a wall}

For a flow constrained in $\mathbb{R}_{+}^{d}$, i.e. for a flow past
a flat plate, the velocity field $u(x,t)=(u^{1}(x,t),\cdots,u^{d}(x,t))$
is extended to the whole space $\mathbb{R}^{d}$ by reflection about
$x_{d}=0$ so that the relation $u(\bar{x},t)=\overline{u(x,t)}$
is maintained for all $t$ and $x$. Then the filtered velocity is
defined to be
\[
\tilde{u}(x,t)=\int_{\mathbb{R}^{d}}\chi(x-y)u(y,t)\mathrm{d}y,
\]
for every $x\in\mathbb{R}^{d}$ and $t$. Therefore, by a change of
variable, 
\begin{equation}
\tilde{u}^{i}(x,t)=\int_{\mathbb{R}_{+}^{d}}\chi_{+}(x,y)u^{i}(y,t)\mathrm{d}y,\quad\tilde{u}^{d}(x,t)=\int_{\mathbb{R}_{+}^{d}}\chi_{-}(x,y)u^{d}(y,t)\mathrm{d}y,\label{av-qq1}    
\end{equation}
for $i=1,\cdots,d-1$, and any $x$ and $t$, where
\begin{equation}
\chi_{\pm}(x,y)=\chi(x-y)\pm\chi(x-\bar{y}),\label{plus-filter}
\end{equation}
for any $x,y\in\mathbb{R}^{d}$. In other words 
\[
\tilde{u}^{i}(x,t)=\int_{\mathbb{R}_{+}^{d}}\chi_{i}(x,y)u^{i}(y,t)\mathrm{d}y,
\]
for all $i=1,\cdots,d$, where $\chi_{i}=\chi_{+}$ for $i=1,\cdots,d-1$
and $\chi_{d}=\chi_{-}$.
\begin{lem}
For a flow in $\mathbb{R}_{+}^{d}$, it holds that
\begin{align}
\tilde{u}^{i}(x,t) & =\int_{\mathbb{R}_{+}^{d}}\left(\mathbb{E}\left[\chi_{i}(x,X_{t}^{\eta})1_{\left\{ t<\zeta\right\} }\right]-\mathbb{E}\left[\chi_{i}(x,X_{t}^{\bar{\eta}})1_{\left\{ t<\zeta\right\} }\right]\right)u_{0}^{i}(\eta)\textrm{d}\eta \nonumber \\
 & \;+\int_{0}^{t}\int_{\mathbb{R}_{+}^{d}}\mathbb{E}\left[1_{\left\{ t<\zeta,t-s<\zeta\circ\tau_{t}\right\} }\chi_{i}(x,X_{t}^{\eta})g^{i}(X_{s}^{\eta},s)\right]\textrm{d}\eta\textrm{d}s,
\end{align}
for $x\in\mathbb{R}_{+}^{d}$ and $t>0$, $\chi_{i}=\chi_{+}$ for
$i=1,\cdots,d-1$ and $\chi_{d}=\chi_{-}$. Here $\tau_{t}$ is the
time reverse operator from $t>0$ and $\zeta=\inf\left\{ t>0:X_{t}\notin\mathbb{R}_{+}^{d}\right\} $.
\end{lem}

\begin{proof}
According to Theorem \ref{thmflowpastw}, 
\begin{align*}
u^{i}(y,t) & =\int_{\mathbb{R}_{+}^{d}}\left(p_{u}(0,\eta;t,y)-p_{u}(0,\bar{\eta};t,y)\right)u_{0}^{i}(\eta)\textrm{d}\eta\nonumber \\
 & +\int_{0}^{t}\int_{\mathbb{R}_{+}^{d}}\mathbb{E}\left[\left.1_{\{t-s<\zeta(X^{\eta}\circ\tau_{t})\}}g^{i}(X_{s}^{\eta},s)\right|X_{t}^{\eta}=\xi\right]p_{u}(0,\eta;t,y)\textrm{d}\eta\textrm{d}s,
\end{align*}
for all $i=1,\cdots,d$, where $\tau_{t}$ is the time reverse from
$t$. By using the Fubini theorem, we have the following integral
representation.
\begin{align*}
&\int_{\mathbb{R}_{+}^{d}}\chi_{i}(x,y)u^{i}(y,t)  \\ &=\int_{\mathbb{R}_{+}^{d}}\left(\int_{\mathbb{R}_{+}^{d}}\chi_{i}(x,y)p_{u}(0,\eta;t,y)\textrm{d}y-\int_{\mathbb{R}_{+}^{d}}\chi_{i}(x,y)p_{u}(0,\bar{\eta};t,y)\textrm{d}y\right)u_{0}^{i}(\eta)\textrm{d}\eta\nonumber \\
 & +\int_{0}^{t}\int_{\mathbb{R}_{+}^{d}}\int_{\mathbb{R}_{+}^{d}}\chi_{i}(x,y)\mathbb{E}\left[\left.1_{\{t-s<\zeta(X^{\eta})\}}g^{i}(X_{s}^{\eta},s)\right|X_{t}^{\eta}=\xi\right]p_{u}(0,\eta;t,y)\textrm{d}y\textrm{d}\eta\textrm{d}s\label{u-int1-1-1}\\
 & =\int_{\mathbb{R}_{+}^{d}}\left(\mathbb{E}\left[\chi_{i}(x,X_{t}^{\eta})1_{\left\{ t<\zeta\right\} }\right]-\mathbb{E}\left[\chi_{i}(x,X_{t}^{\bar{\eta}})1_{\left\{ t<\zeta\right\} }\right]\right)u_{0}^{i}(\eta)\textrm{d}\eta\\
 & +\int_{0}^{t}\int_{\mathbb{R}_{+}^{d}}\mathbb{E}\left[1_{\left\{ t<\zeta\right\} }1_{\{t-s<\zeta\circ\tau_{t}\}}\chi_{i}(x,X_{t}^{\eta})g^{i}(X_{s}^{\eta},s)\right]\textrm{d}\eta\textrm{d}s,
\end{align*}
which completes the proof.
\end{proof}
This representation suggests the following approximation.

\vskip0.3truecm

\begin{itemize}
    \item \textbf{\emph{Numerical schemes for random LES method for flows past walls}}
\end{itemize}

\vskip0.3truecm

Run a diffusion with vector field $U(x,t)$:

\begin{equation}
\textrm{d}Y_{t}^{\eta}=U(Y_{t}^{\eta},t)\textrm{d}t+\sqrt{2\nu}\textrm{d}B_{t},\quad Y_{0}^{\eta}=\eta,\label{half1}
\end{equation}
for every $\eta\in\mathbb{R}^{d}$, where $B=(B^{1},\cdots,B^{d})$
is a standard Brownian motion on a probability space $(\varOmega,\mathcal{F},\mathbb{P})$;
\begin{align}
U^{i}(x,t) & =\int_{\mathbb{R}_{+}^{d}}\left(\mathbb{E}\left[\chi_{i}(x,Y_{t}^{\eta})1_{\left\{ t<\zeta\right\} }\right]-\mathbb{E}\left[\chi_{i}(x,Y_{t}^{\bar{\eta}})1_{\left\{ t<\zeta\right\} }\right]\right)u_{0}^{i}(\eta)\textrm{d}\eta\nonumber \\
 & \;+\int_{0}^{t}\int_{\mathbb{R}_{+}^{d}}\mathbb{E}\left[1_{\left\{ t<\zeta,t-s<\zeta\circ\tau_{t}\right\} }\chi_{i}(x,Y_{t}^{\eta})G^{i}(Y_{s}^{\eta},s)\right]\textrm{d}\eta\textrm{d}s,\label{half2}
\end{align}
for $x\in\mathbb{R}_{+}^{d}$, $i=1,\cdots,d$, and 
\begin{equation}
U(x,t)=\overline{U(\bar{x},t)}\quad\textrm{ for }x\in\mathbb{R}_{-}^{d},\quad U(x,t)=0\quad\textrm{ if }x_{d}=0,\label{half3}
\end{equation}
where $u_{0}$ is the initial velocity;
\begin{equation}
G(x,t)=-\nabla P(x,t)+F(x,t)\quad\textrm{ for }x\in\mathbb{R}_{+}^{d},t\geq0,\label{half4}
\end{equation}
where $F=(F^{1},F^{2},F^{3})$ is an external force applying to the
flow, and

(1) if $d=2$,
\begin{align}
\nabla P(x,t) & =\int_{\mathbb{R}_{+}^{2}}K_{2}^{+}(x,y)\left.\left(\nabla\cdot F-\sum_{i,j=1}^{2}\frac{\partial U^{j}}{\partial y^{i}}\frac{\partial U^{i}}{\partial y^{j}}\right)\right|_{(y,t)}\textrm{d}y\nonumber \\
 & \;+\int_{\mathbb{R}}K_{2}^{+}(x,(y_{1},0))F^{2}(y_{1},0)\textrm{d}y_{1}\nonumber \\
 & \;+\nu\int_{\mathbb{R}}\frac{\partial K_{2}^{+}}{\partial y_{1}}(x,(y_{1},0))\frac{\partial U^{1}}{\partial y_{2}}(y_{1},0)\textrm{d}y_{1},\label{half5-2D}
\end{align}

(2) if $d=3$,
\begin{align}
\nabla P(x,t) & =\int_{\mathbb{R}_{+}^{3}}K_{3}^{+}(x,y)\left.\left(\nabla\cdot F-\sum_{i,j=1}^{3}\frac{\partial U^{j}}{\partial y^{i}}\frac{\partial U^{i}}{\partial y^{j}}\right)\right|_{(y,t)}\textrm{d}y\nonumber \\
 & \;+\int_{\mathbb{R}^{2}}K_{3}^{+}(x,(y_{1},y_{2},0))F^{3}(y_{1},y_{2},0)\textrm{d}y_{1}\textrm{d}y_{2}\nonumber \\
 & \;+\nu\int_{\mathbb{R}^{2}}\frac{\partial K_{3}^{+}}{\partial y_{1}}(x,(y_{1},y_{2},0))\frac{\partial U^{1}}{\partial y_{3}}(y_{1},y_{2},0)\textrm{d}y_{1}\textrm{d}y_{2}\nonumber \\
 & \;+\nu\int_{\mathbb{R}^{2}}\frac{\partial K_{3}^{+}}{\partial y_{2}}(x,(y_{1},y_{2},0))\frac{\partial U^{2}}{\partial y_{3}}(y_{1},y_{2},0)\textrm{d}y_{1}\textrm{d}y_{2},\label{half3D}
\end{align}
for $x\in\mathbb{R}_{+}^{d}$ and $t\geq0$.

\section{Numerical experiments and analysis}
In this section, we present the numerical schemes in detail. Subsequently, we conduct comprehensive experiments to rigorously evaluate the effectiveness of the Random LES method. Finally, we provide an in-depth analysis of the proposed approach and perform systematic comparisons with existing methods, highlighting several distinct advantages of our technique.

\subsection{Numerical schemes and variables settings}

In this section, we present comprehensive numerical experiments based
on the scheme defined in Section 3. We are able to present
the numerical scheme based on above formulas, and we use the Euler method to
discrete the integral. The initial velocity $U(x,0)$ and external
force $F$ is given, so the $\nabla\cdot F$ can be calculated directly.

Set mesh size $s>0$, time step $\delta>0$ and kinematic viscosity
$\nu>0$. For $i_{1},i_{2},i_{3}\in\mathbb{Z}$, denote $y^{i_{1,}i_{2},i_{3}}=(i_{1},i_{2},i_{3})s$,
$U^{i_{1},i_{2},i_{3}}=U(y^{i_{1,}i_{2},i_{3}},0)$. In numerical
scheme, we drop the expectation and use one-copy of Brownian particles.
We can discrete the stochastic differential equation (\ref{half1}) following
Euler scheme: for $t_{i}=i\delta,i=0,1,2,\cdots$.

\begin{equation}
    Y_{t_{k}}^{i_{1},i_{2},i_{3}}=Y_{t_{k-1}}^{i_{1},i_{2},i_{3}}+\delta U(Y_{t_{k-1}}^{i_{1},i_{2},i_{3}},t_{k-1})+\sqrt{2\nu}(B_{t_{k}}-B_{t_{k-1}}),\quad Y_{0}^{i_{1},i_{2},i_{3}}=y^{i_{1,}i_{2},i_{3}},
\end{equation}

where $Y_{t_{k}}^{i_{1},i_{2},i_{3}}$ denotes $Y_{t_{k}}^{x^{i_{1},i_{2},i_{3}}}$
for simplicity. However, the forward Euler method is not considered as an advanced and accurate enough method. To better calculate the stochastic differential equation, we use the Milstein method \citep{Milstein1975}. The Milstein method takes into account the derivative of the diffusion term and uses the It\^o's lemma expansion to correct the discretization error, the scheme is :
\begin{align}
    Y_{t_{k}}^{i_{1},i_{2},i_{3}} &=Y_{t_{k-1}}^{i_{1},i_{2},i_{3}}+\delta U(Y_{t_{k-1}}^{i_{1},i_{2},i_{3}},t_{k-1})+\sqrt{2\nu}(B_{t_{k}}-B_{t_{k-1}}) \nonumber\\
    &+\frac{1}{2}U(Y_{t_{k-1}}^{i_{1},i_{2},i_{3}},t_{k-1})U'(Y_{t_{k-1}}^{i_{1},i_{2},i_{3}},t_{k-1})[(B_{t_{k}}-B_{t_{k-1}})^2-\delta].
\end{align}

The $U'$ can be calculated explicitly. Then integral representations in (\ref{half2}) and
(\ref{half3D}) can be discretized as follows: 
\begin{align}
U(x,t_{k}) & =\sum_{i_{1},i_{2}\in \mathbb{R},i_{3}>0}s^{3}\left[1_{\mathbb{R}_{+}^{3}}\left(Y_{t_{l}}^{i_{1},i_{2},i_{3}}\right)\chi_i(x-Y_{t_{k}}^{i_{1},i_{2},i_{3}})\right]U_0^{i_{1},i_{2},i_{3}}\nonumber \\
& -\sum_{i_{1},i_{2}\in \mathbb{R},i_{3}>0}s^{3}\left[1_{\mathbb{R}_{+}^{3}}\left(Y_{t_{l}}^{i_{1},i_{2},-i_{3}}\right)\chi_i(x-Y_{t_{k}}^{i_{1},i_{2},-i_{3}})\right]U_0^{i_{1},i_{2},i_{3}}\nonumber \\
 & +\sum_{i_{1},i_{2}\in \mathbb{R},i_{3}>0}\sum_{j=1}^{k}s^{3}\delta\prod_{l=j}^{k}1_{\mathbb{R}_{+}^{3}}\left(Y_{t_{l}}^{i_{1},i_{2},i_{3}}\right)\chi_i(x-Y_{t_{k}}^{i_{1},i_{2},i_{3}})G(X_{t_{j-1}}^{i_{1},i_{2},i_{3}},t_{j-1}),
\end{align}
for $x\in\mathbb{R}_{+}^d$, and 
\begin{equation}
U(x,t)=\overline{U(\bar{x},t)}\quad\textrm{ for }x\in\mathbb{R}_{-}^{d},\quad U(x,t)=0\quad\textrm{ for } x_3=0,
\end{equation}
where $G(x,t_{k})=-\nabla P(x,t_{k})+F(x,t_{k})$,
\begin{align}
\nabla P(x,t_k) & =\sum_{i_{1},i_{2}\in \mathbb{R},i_{3}>0}K_{3}^{+}(x,(i_1,i_2,i_3))\nabla\cdot F((i_1,i_2,i_3),t_k)\nonumber \\
& -\sum_{i_{1},i_{2}\in \mathbb{R},i_{3}>0}K_{3}^{+}(x,(i_1,i_2,i_3))\sum_{i,j=1}^{3}\frac{\partial U^{j}}{\partial x^{i}}((i_1,i_2,i_3),t_k)\frac{\partial U^{i}}{\partial x^{j}}((i_1,i_2,i_3),t_k)\nonumber \\
 & \;+\sum_{i_{1},i_{2}\in \mathbb{R}}s^2K_{3}^{+}(x,(i_{1},i_{2},0))F^{3}(i_{1},i_{2},0)\nonumber \\
 & \;+\nu\sum_{i_{1},i_{2}\in \mathbb{R}}s^2\frac{\partial K_{3}^{+}}{\partial y_{1}}(x,(i_{1},i_{2},0))\frac{\partial U^{1}}{\partial x_{3}}((i_{1},i_{2},0),t_k)\nonumber \\
 & \;+\nu\sum_{i_{1},i_{2}\in \mathbb{R}}s^2\frac{\partial K_{3}^{+}}{\partial y_{2}}(x,(i_{1},i_{2},0))\frac{\partial U^{2}}{\partial x_{3}}((i_{1},i_{2},0),t_k),
\end{align}
for $x=(x_{1},x_{2},x_{3})$ where $x_{3}>0$, the $\frac{\partial U^{j}}{\partial x^{i}}$ can be explicitly represented and calculated directly. Each variable is updated based on its value from the previous time step, the iterative computational procedure is as follows:
\begin{align*}
    & U_0,Y_0 \  \rightarrow  Y_{1} \xrightarrow{+G_0,U_0} U_1 \xrightarrow{+\frac{\partial U_1}{x},F}G_1,\\
    & U_{1},Y_{1} \rightarrow Y_{2}\xrightarrow{+G_1,U_0} U_2 \xrightarrow{+\frac{\partial U_2}{x},F}G_2 ,\\
    & \cdots \cdots \\
    & U_{t_{n-1}},Y_{n-1} \rightarrow Y_{T}\xrightarrow{+G_{t_{n-1}},U_0} U_T \xrightarrow{+\frac{\partial U_T}{x},F}G_T. \\
\end{align*}

To validate the correctness of the theory and numerical methods, we
conducted extensive numerical experiments. For comprehensiveness,
the following experiments included scenarios with both laminar flows
and turbulence, in two-dimensional and three-dimensional settings,
as well as in unbounded domains and with wall bounded space. From
the perspective of numerical simulation, turbulence represents a more
challenging flow regime to model. By comparing the results of laminar
flow experiments with those of turbulent flow experiments, we can
more clearly discern the differences between these two flow states.
The results demonstrate that random LES method can accurately and
efficiently simulate fluid systems.

Now we introduce some fundamental settings. Reynolds number $Re$
is defined by $Re=\frac{U_{0}L}{\nu}$, where $U_{0}$ is the main
stream velocity, $\nu$ is kinematic viscosity and $L$ is the length
scale. In unbounded experiments, the observation region is set as
a square, meaning that the horizontal length scale ($L_{h}$) and
vertical length scale ($L_{v}$) of the observation area are equal.
For the space with solid wall boundary, a smaller observation domain
and a more refined grid discretization are employed in the vertical
direction to enhance the resolution of boundary phenomena. The mesh
size $s\thicksim L\sqrt{\frac{1}{Re}}$, we set the vertical mesh
size as $s_{v}$ and the horizontal mesh size as $s_{h}$ and the
external force as $F$.

\subsection{Experiments for random LES method}

According to the numerical scheme above, we conducted four sets of experiments to simulate four distinct flow scenarios by setting different initial conditions and external forces.
\begin{itemize}
    \item \textbf{Experiment 1:} 2-dimensional laminar advective flow through a solid wall with external force 
    \item \textbf{Experiment 2:}  2-dimensional turbulent advective flow through a solid wall with external force 
    \item \textbf{Experiment 3:}  3-dimensional flow of an initially quiescent fluid driven by tensile forces. 
    \item \textbf{Experiment 4:} 3-dimensional turbulent advective flow through a solid wall with external force 
\end{itemize}
\subsubsection{Experiment 1: 2-dimensional advection through a wall (laminar flows)}

\label{l2b} Now we consider the bounded case. In this part, $\nu=0.3$,
$Re=1000$, length scale $L_{h}=3\pi,L_{v}=0.2\pi$, so that $U_{0}=\frac{\nu}{L_{h}}Re=31.83$.
The mesh size $s_{h}=\frac{3\pi}{50},s_{v}=\frac{0.2\pi}{50}$. and
the time step $\delta=0.001$. The numerical experiment is demonstrated
at times $t=0.03,t=0.06,t=0.09$. We set the initial velocity to be
of the form $U(x,0)=(U_{0},0)$, and set force $F=(10e^{(-\frac{(x_{1}^{2}+x_{2}^{2})}{2s_{h}s_{v}})},-9.81)$.
The velocity field and the vorticity field are shown in Figure \ref{l_bounded_2du} and \ref{l_bounded_2dw}.

\begin{figure}[H]
\centering
\includegraphics[width=1\textwidth]{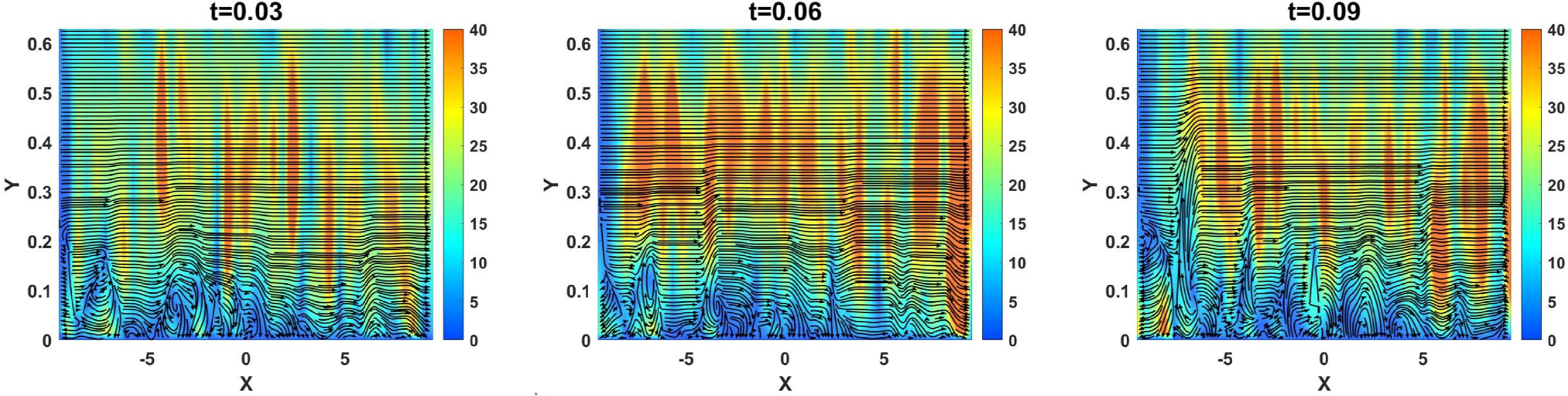}
\centering \caption{\foreignlanguage{english}{Velocity fields of wall-bounded laminar
flows on $\mathbb{R}_{+}^{2}$}}
\label{l_bounded_2du} 
\end{figure}
\begin{figure}[H]
\centering
\includegraphics[width=1\textwidth]{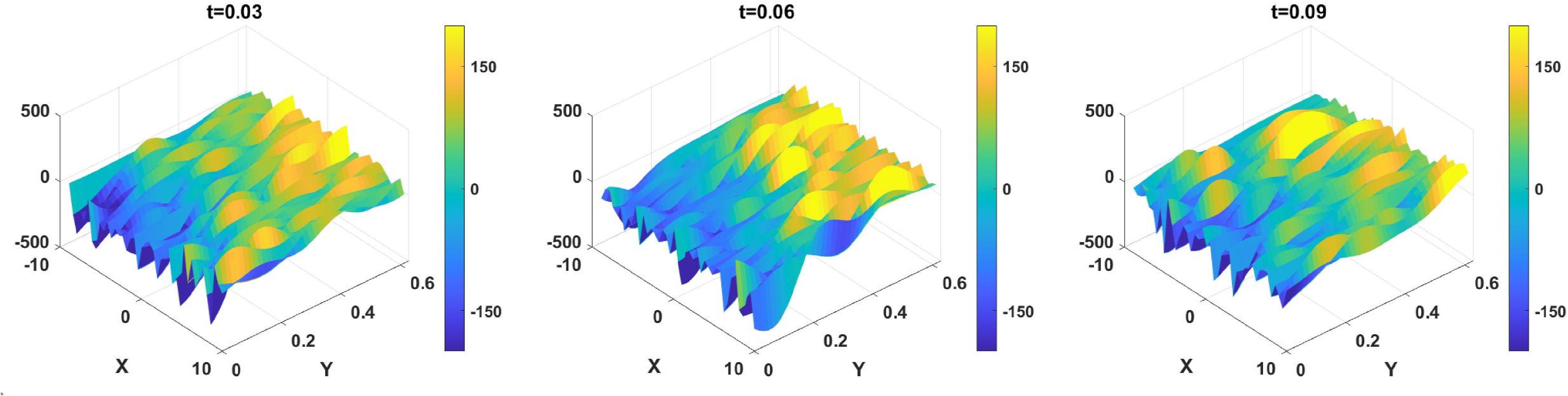}
\centering \caption{\foreignlanguage{english}{Vorticity of wall-bounded laminar
flows on $\mathbb{R}_{+}^{2}$}}
\label{l_bounded_2dw} 
\end{figure}
As shown in Figure \ref{l_bounded_2du}, the fluid system initially
remains in a relatively stable state. Over time, a small number of
vortices begin to emerge at the boundary and gradually propagate into
the interior of the domain. This phenomenon is primarily driven by
shear forces, specifically the combined effects of boundary pressure
and nonlinear forces. Notably, compared to the subsequent turbulent
flow experiments (Figure \ref{t_bounded_2du}), the fluid dynamics
in the laminar flow experiments exhibit significantly lower intensity,
with fewer vortices generated. 

In laminar or low Reynolds number case, where viscous forces prevail over fluid inertia, the vorticity dynamics are characterized by diffusion and wall confinement. The no-slip condition at the wall ($y=0$) acts as a source of strong vorticity. However, its transport into the flow interior occurs primarily through slow molecular diffusion, leading to rapid attenuation of strength with distance from the wall. The observed incomplete vortices are, in essence, starting or separation vortices formed by the rolling-up of this diffused vorticity layer under the action of external shear. The fundamental reason for their distorted morphology is twofold: the concentration of vorticity is insufficient, and the rotational energy of the nascent vortex is too weak to overcome the powerful viscous dissipation and the anchoring effect of the wall. Consequently, these structures remain tethered to the boundary, unable to lift off and coalesce into a closed, symmetric vortex core. Their vorticity is primarily sustained by the local wall shear rate rather than by large-scale vortex entrainment. This observation is consistent with established theoretical principles.

\subsubsection{Experiment 2: 2-dimensional advection through a wall (turbulent flows)}

\label{t2b} In this subsection, we present a numerical experiment
simulating a two-dimensional turbulent flow on $\mathbb{R}_{+}^{2}$,
where the flow passes over a solid wall at a high Reynolds number.
To better capture the boundary effects, we use a smaller observation
domain and finer spatial grid resolution in the vertical direction.
The parameters for the experiment are as follows: $\nu=0.3$, $Re=5500$,
horizontal length scale $L_{h}=3\pi$, vertical length scale $L_{v}=0.2\pi$
and the initial main stream velocity $U_{0}=135.15$. The horizontal
mesh size $s_{h}=\frac{3\pi}{50}$, vertical mesh size $s_{v}=\frac{0.2\pi}{50}$
and the time step $\delta=0.0001$. The numerical results are demonstrated
at three time instances: $t=0.01,t=0.02,t=0.03$. The initial velocity
field is defined as $U(x,0)=(U_{0},0)$ while the external force is
given by $F=(100e^{(-\frac{(x_{1}^{2}+x_{2}^{2})}{2s_{h}s_{v}})},-9.81)$.
The velocity and vorticity are shown in Figure \ref{t_bounded_2du} and \ref{t_bounded_2dw}.

\begin{figure}[H]
\centering
\includegraphics[width=1\textwidth]{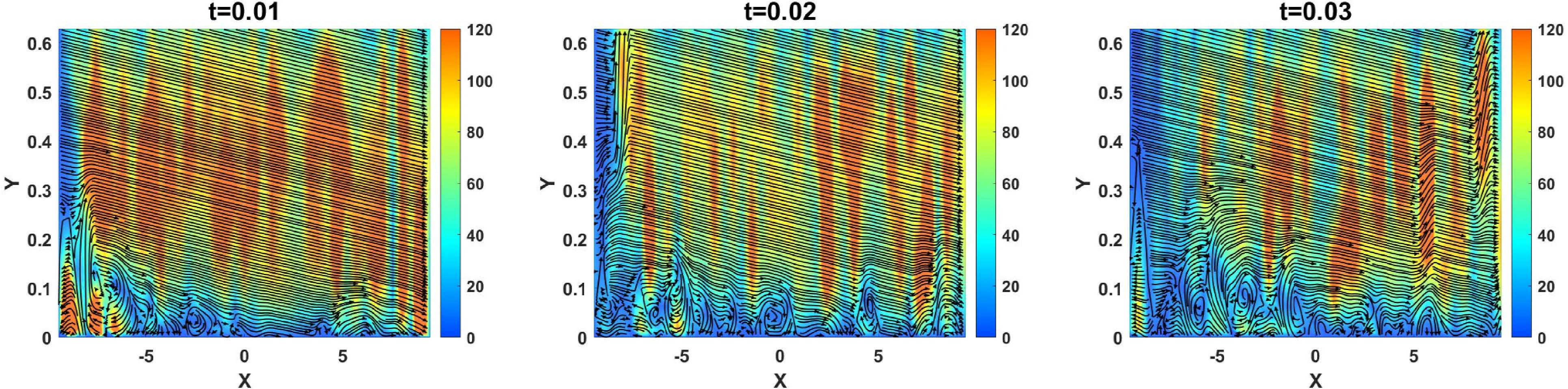} \caption{\foreignlanguage{english}{Velocity fields of wall-bounded turbulent
flows on $\mathbb{R}_{+}^{2}$}}
\label{t_bounded_2du} 
\end{figure}
\begin{figure}[H]
\centering
\includegraphics[width=1\textwidth]{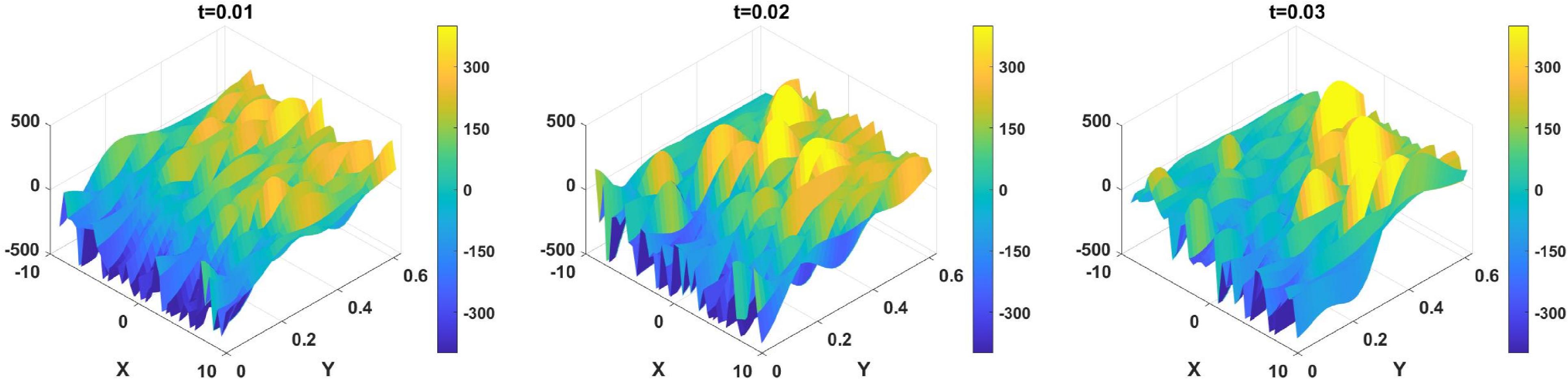} \caption{\foreignlanguage{english}{Vorticity of wall-bounded turbulent
flows on $\mathbb{R}_{+}^{2}$}}
\label{t_bounded_2dw} 
\end{figure}
As shown in Figure \ref{t_bounded_2du} and \ref{t_bounded_2dw}, the fluid system initially remains in a relatively stable state. Over time, vortices begin to
form at the boundary and gradually propagate into the interior of
the domain. This phenomenon is primarily attributed to shear forces,
specifically the combined influence of boundary pressure and nonlinear
forces. Additionally, the vorticity near the boundary exhibits significantly higher values, which is consistent with physical observations in real-world fluid flows.

In the turbulent regime, inertial effects dominate. The wall remains the primary source of vorticity, but turbulence greatly enhances its production and transport through stretching and reorientation mechanisms. Small-scale turbulent fluctuations tear and lift the wall-generated vorticity layers, rolling them into concentrated vortex tubes or packets. Compared with laminar flows (Experiment 1), the complete vortices observed in turbulent flow owe their integrity to the sufficient kinetic energy contained within their rotational cores, which enables them to lift away from the wall and extend into the outer flow region. These vortices form relatively closed streamline loops with clearly identifiable low-pressure cores, indicating a well-developed and coherent vortex structure. The larger magnitude of vorticity near the turbulent boundary does not arise from an infinitely increasing rate of wall generation, but rather from the intensified transport, accumulation, and amplification of vorticity driven by turbulent motions.

\subsubsection{Experiment 3: 3-dimensional initially stationary fluid}

In this subsection, we set an interesting 3-dimensional experiment. We assume that the flow is stationary initially (which is actually a laminar flow), then a tensile force in the same direction was applied to the fluid, observing the changes in the fluid system. We set
$\nu=0.3$, length scale $L=3\pi$, the mesh size $s=\frac{3\pi}{50}\thicksim L\sqrt{\frac{1}{Re}}$
and the time step $\delta=0.001$. The numerical experiment is demonstrated
at times $t=0.1,t=0.2,t=0.3$. We set the initial velocity to be
of the form $U(x,0)=(0,0,0)$, and
set force $F=(100e^{(-\frac{(x_{1}^{2}+x_{2}^{2}+x_{3}^{2})}{2s_{h}^{2}s_{v}})},100e^{(-\frac{(x_{1}^{2}+x_{2}^{2}+x_{3}^{2})}{2s_{h}^{2}s_{v}})},-9.81)$. The velocity field and the vorticity
field are shown in Figure \ref{a_bounded_3du} and \ref{a_bounded_3dw}.

\begin{figure}[H]
\centering
\includegraphics[width=1\textwidth]{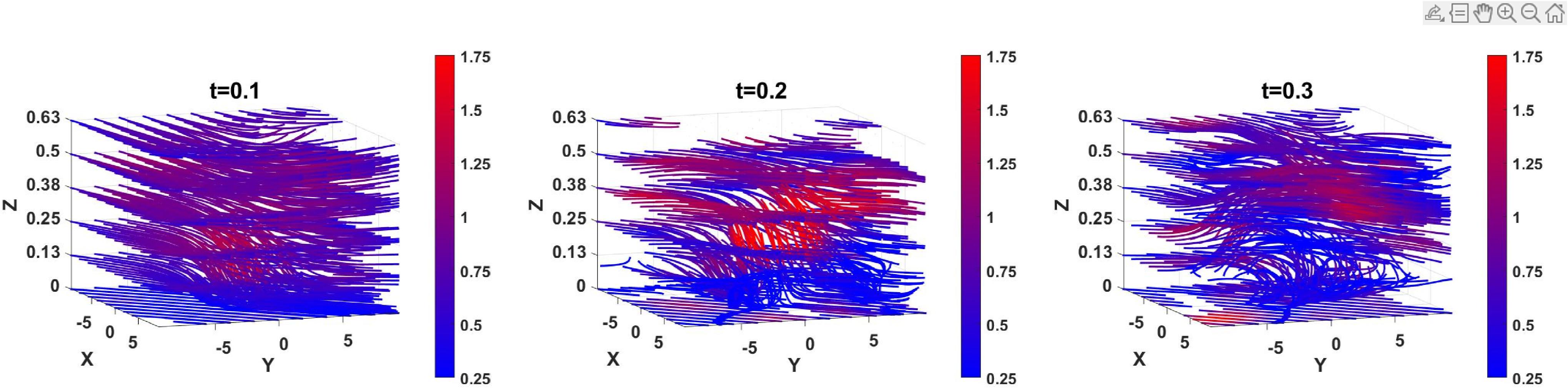}

\caption{\foreignlanguage{english}{Velocity fields wall-bounded
flows on $\mathbb{R}^{3}_{+}$}}
\label{a_bounded_3du} 
\end{figure}

\begin{figure}[H]
\centering
\includegraphics[width=1\textwidth]{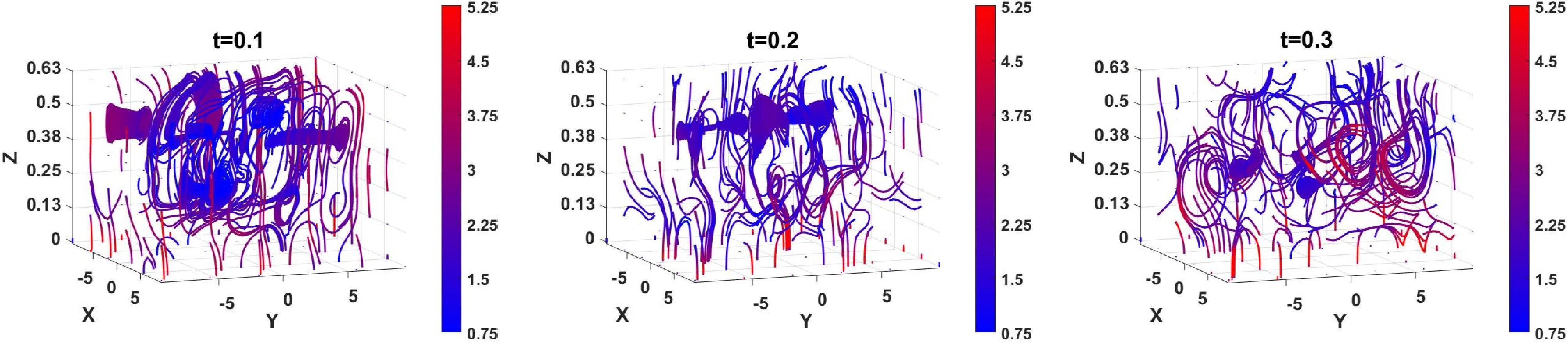}

\caption{\foreignlanguage{english}{Vorticity of wall-bounded
flows on $\mathbb{R}^{3}_{+}$}}
\label{a_bounded_3dw} 
\end{figure}

Figure \ref{a_bounded_3du} shows that the initially stationary fluid system flows along the direction of the pulling force, and as time progresses, vortices are generated at the boundaries. When a tensile force is suddenly applied to an initially stationary (laminar) fluid, the system experiences a strong streamwise acceleration that establishes a velocity gradient between the bulk and the no-slip boundaries. This gradient immediately generates vorticity at the walls, where the velocity remains zero. The wall-bounded vorticity diffuses into the fluid and is simultaneously advected by the emerging mean flow. As time evolves, the non-uniform acceleration induces spatial variations in the strain rate, producing regions of local vorticity amplification. The interaction between the streamwise stretching and the wall-generated shear promotes vortex formation near the boundaries, as observed in figure \ref{a_bounded_3du}. These vortices form from the partial roll-up of wall-generated vorticity layers under the applied extensional flow. Due to the relatively low flow velocity, the rotational energy is limited, resulting in irregular and asymmetric vortex shapes.

In order to better depict the shape of the flow, we make two-dimensional velocity field slices to visually represent the behavior of the fluid. The sectional views are presented  at the planes $x=C$, $y=C$, and $z=C$ in Figures \ref{a_bounded_3d_x}, \ref{a_bounded_3d_y}, and \ref{a_bounded_3d_z}, respectively, with $C=0.1\pi$ chosen for illustration.

\begin{figure}[H]
\includegraphics[width=1\textwidth]{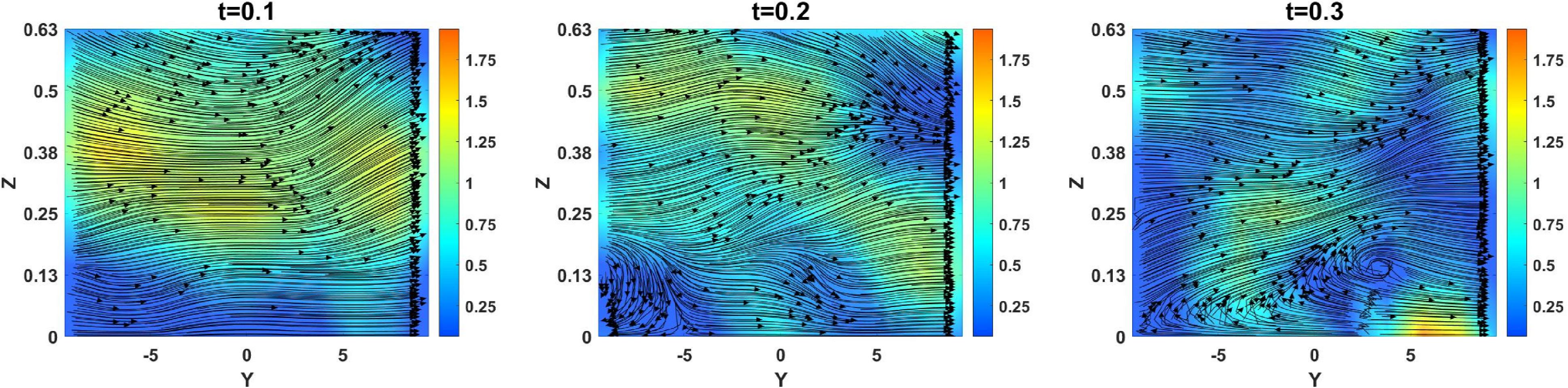}
\caption{\foreignlanguage{english}{Section Velocity fields of incompressible viscous flows on plane 
$x=C$}}
\label{a_bounded_3d_x} 
\end{figure}

\begin{figure}[H]
\includegraphics[width=1\textwidth]{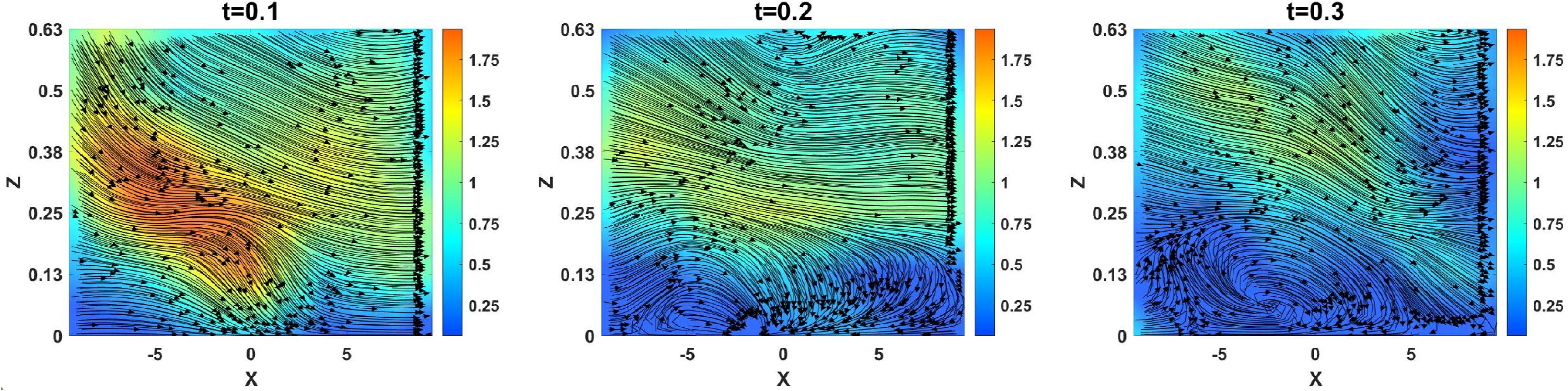}
\caption{\foreignlanguage{english}{Section Velocity fields of incompressible viscous flows on plane 
$y=C$}}
\label{a_bounded_3d_y} 
\end{figure}

\begin{figure}[H]
\includegraphics[width=1\textwidth]{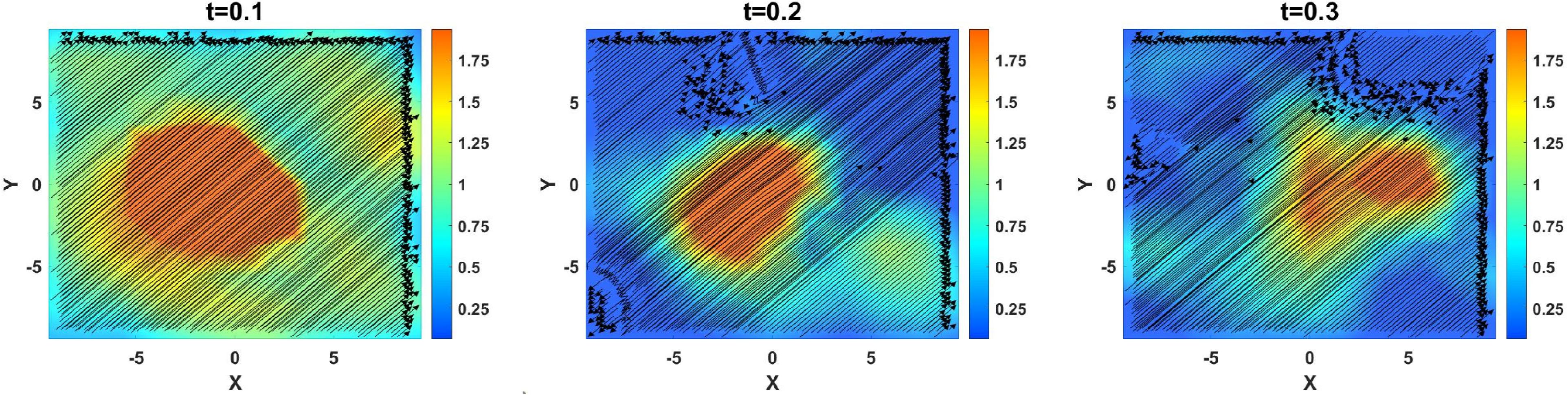}
\caption{\foreignlanguage{english}{Section Velocity fields of incompressible viscous flows on plane
$z=C$}}
\label{a_bounded_3d_z} 
\end{figure}

With the tensile force applied only in the $xy$-plane, the fluid accelerates primarily along $x$ and $y$, generating velocity gradients near the wall at $z=0$. These vertical gradients produce vorticity oriented along the $z$-direction, while horizontal slices in the $xy$-plane show little to no vortex activity. As the flow develops, the wall-generated vorticity is stretched and deformed along the $xy$-plane, forming irregular, elongated vortices confined near the boundaries, reflecting the anisotropic forcing and limited vertical transport, which is consistent with the expected physical behavior.

\subsubsection{Experiment 4: 3-dimensional advection through a wall (turbulent flows)}

\label{l3b} In order to further demonstrate the effectiveness of our method, we complete the experiment of initial advection turbulent advection passing through a wall. In this experiment, $\nu=0.3$, $Re=5500$, $L_{h}=3\pi,L_{v}=0.2\pi$,
so that $U_{0}=\frac{\nu}{L}Re=135$. The horizontal mesh size $s_{h}=\frac{3\pi}{50}$,
vertical mesh size $s_{v}=\frac{0.2\pi}{50}$ , and the time step
$\delta=0.001$. The numerical experiment is demonstrated at times
$t=0.1,t=0.2,t=0.3$, and we also show the structure of flows in $t=0.01$ to see the behavior of the fluid in the initial stage. We set the initial velocity to be of the
form $U(x,0)=(U_{0},U_{0},0)$, and set force $F=(100e^{(-\frac{(x_{1}^{2}+x_{2}^{2}+x_{3}^{2})}{2s_{h}^{2}s_{v}})},100e^{(-\frac{(x_{1}^{2}+x_{2}^{2}+x_{3}^{2})}{2s_{h}^{2}s_{v}})},-9.81)$.
The velocity field and the vorticity field are shown in Figure \ref{l_bounded_3du} and \ref{l_bounded_3dw}.

\begin{figure}[H]
\centering
\includegraphics[width=1\textwidth]{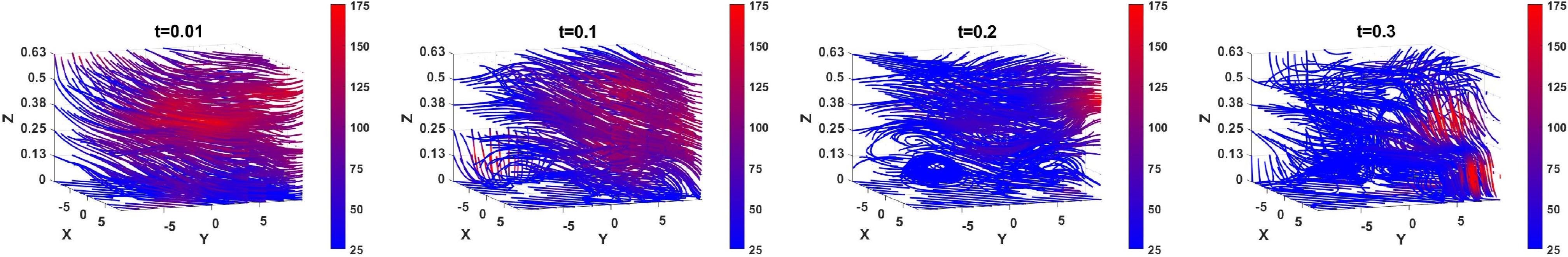}
\caption{\foreignlanguage{english}{Velocity fields of wall-bounded turbulent
flows on $\mathbb{R}_{+}^{3}$}}
\label{l_bounded_3du} 
\end{figure}

\begin{figure}[H]
\centering
\includegraphics[width=1\textwidth]{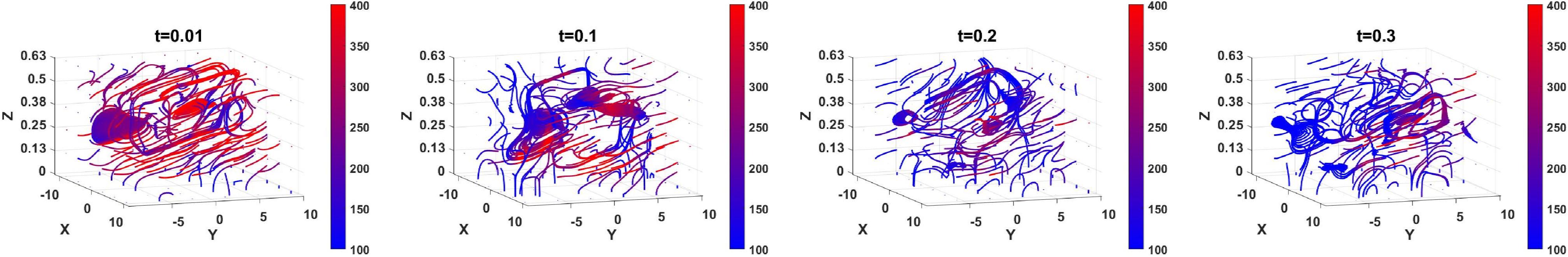}
\caption{\foreignlanguage{english}{Vorticity fields of wall-bounded turbulent
flows on $\mathbb{R}_{+}^{3}$}}
\label{l_bounded_3dw} 
\end{figure}

In this experiment, we specifically present fluid flow visualizations over a short time scale. At early times ($t = 0.01$), the flow remains close to its initial advective state, with negligible vorticity away from the walls. As the flow evolves, shear near the boundary walls generates strong wall-bounded vorticity due to the no-slip condition, which is initially confined close to the wall. This vorticity is amplified and stretched by the developing flow, producing vortices that grow in size and complexity, as illustrated in Figures \ref{l_bounded_3du} and \ref{l_bounded_3dw}.

Compared with laminar flows, inertial effects in the turbulent regime dominate, leading to highly three-dimensional and unsteady dynamics. The enhanced advection and vortex stretching in turbulence increase the local vorticity magnitude near the walls, making boundary-layer vortices more intense than in laminar cases. Small-scale turbulent fluctuations continuously distort and redistribute the wall-generated vorticity, giving rise to irregular, asymmetric vortex structures that penetrate into the outer flow. In effect, the observed vorticity field reflects the combined influence of wall generation, inertial amplification, and three-dimensional turbulent transport. To enhance the clarity of the flow structures, two-dimensional sectional views are shown. 

\begin{figure}[H]
\includegraphics[width=1\textwidth]{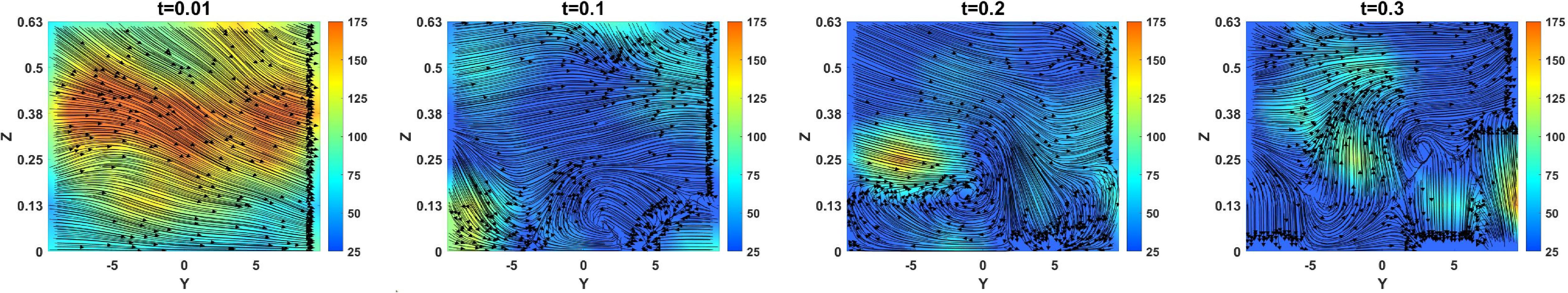}
\caption{\foreignlanguage{english}{Section Velocity fields of incompressible viscous flows on plane 
$x=C$}}
\label{l_bounded_3d_x} 
\end{figure}

\begin{figure}[H]
\includegraphics[width=1\textwidth]{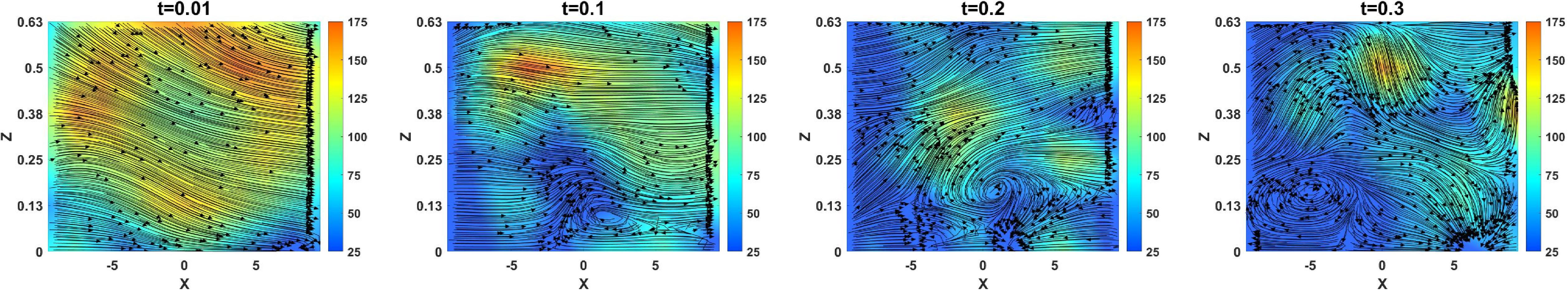}
\caption{\foreignlanguage{english}{Section Velocity fields of incompressible viscous flows on plane 
$y=C$}}
\label{l_bounded_3d_y} 
\end{figure}

\begin{figure}[H]
\includegraphics[width=1\textwidth]{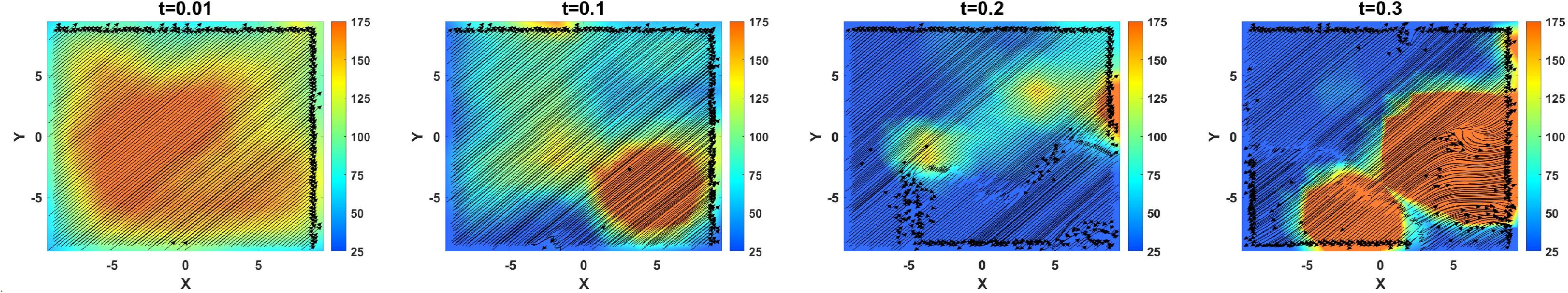}
\caption{\foreignlanguage{english}{Section Velocity fields of incompressible viscous flows on plane 
$z=C$}}
\label{l_bounded_3d_z} 
\end{figure}
In plane $x=C$ (Figure \ref{l_bounded_3d_x}) and $y=C$  (Figure \ref{l_bounded_3d_y}), the overall flow direction of the fluid remains largely aligned with the initial advective direction and the direction of the applied force. As time progresses, fully developed vortices gradually emerge from the boundaries. The explanation of the flow mechanism of the section is similar to that of Experiment 3. 

\subsection{Analysis and comparison}

\subsubsection{Divergence-free analysis}

For 2D flows, the value of divergence is easy for visualization. But for 3D case, the divergence is not easily visualized through 3D images, the absolute value of the divergence at each grid point is calculated and then the spatial average is taken. We print the time series figures of the overall average divergence of the fluid system. Following the conventions in the existing literature, a regular method to verify the divergence-free analysis is to calculate whether $\frac{|\nabla \cdot u|}{|\nabla u|}\ll1$. To achieve this, we define 
$$\text{test}(t)=\frac{1}{(\frac{L_h^2*L_v}{s_h^2*s_v})}\sum_{x\in D} \frac{|\nabla \cdot u(x,t)|}{||\nabla u(x,t)||_2},$$ 
to calculate the mean divergence value of each timestep, the results are shown in Figure \ref{fig:both1} and \ref{fig:both2}.

\begin{figure}[H]
    \centering
    \begin{subfigure}[b]{0.48\textwidth}
        \includegraphics[width=\textwidth]{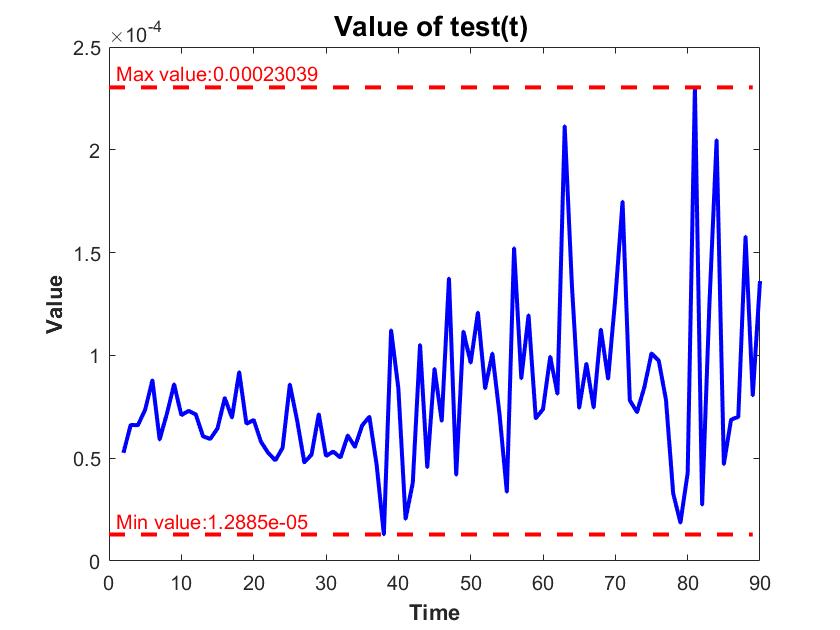}
        \caption{Experiment 1}
        \label{fig:image1}
    \end{subfigure}
    \hfill 
    \begin{subfigure}[b]{0.48\textwidth}
        \includegraphics[width=\textwidth]{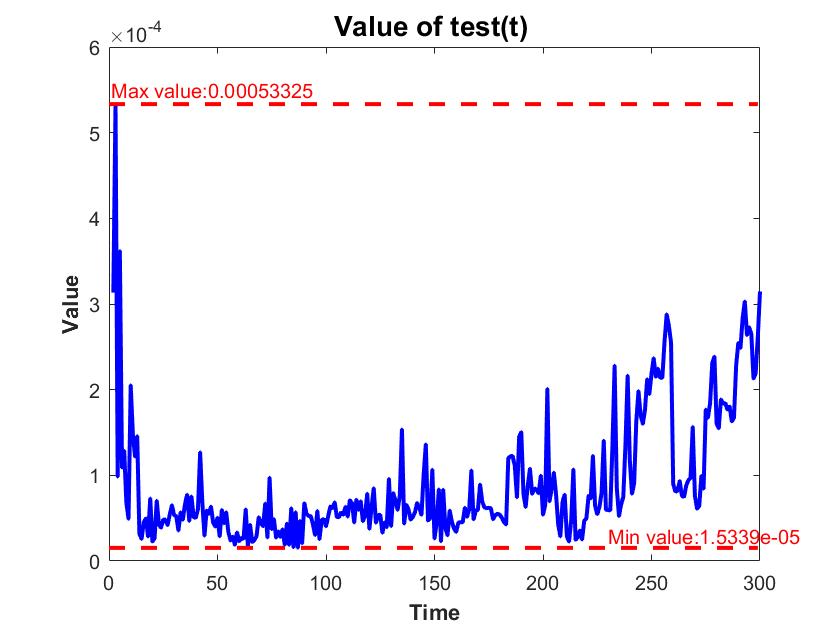}
        \caption{Experiment 2}
        \label{fig:image2}
    \end{subfigure}
    \caption{Incompressible test for Experiment 1 (left) and Experiment 2 (right)}
    \label{fig:both1}
\end{figure}

From Figure \ref{fig:both1}, it can be observed that $\frac{|\nabla \cdot u|}{|\nabla u|}$ is significantly less than 1 in the two-dimensional case. A comparison between Experiment 1 and Experiment 2 reveals that the enforcement of the incompressibility condition is less effective in Experiment 2, owing to the presence of turbulence at high Reynolds numbers. The results indicate that the Random LES method generally preserves the incompressibility assumption.

\begin{figure}[H]
    \centering
    \begin{subfigure}[b]{0.48\textwidth}
        \includegraphics[width=\textwidth]{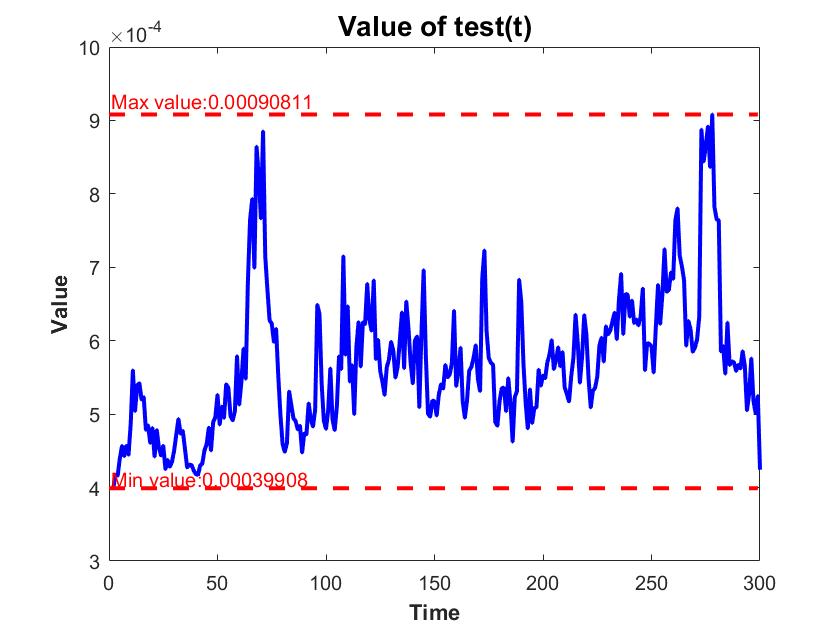}
        \caption{Experiment 3}
        \label{fig:image3}
    \end{subfigure}
    \hfill 
    \begin{subfigure}[b]{0.48\textwidth}
        \includegraphics[width=\textwidth]{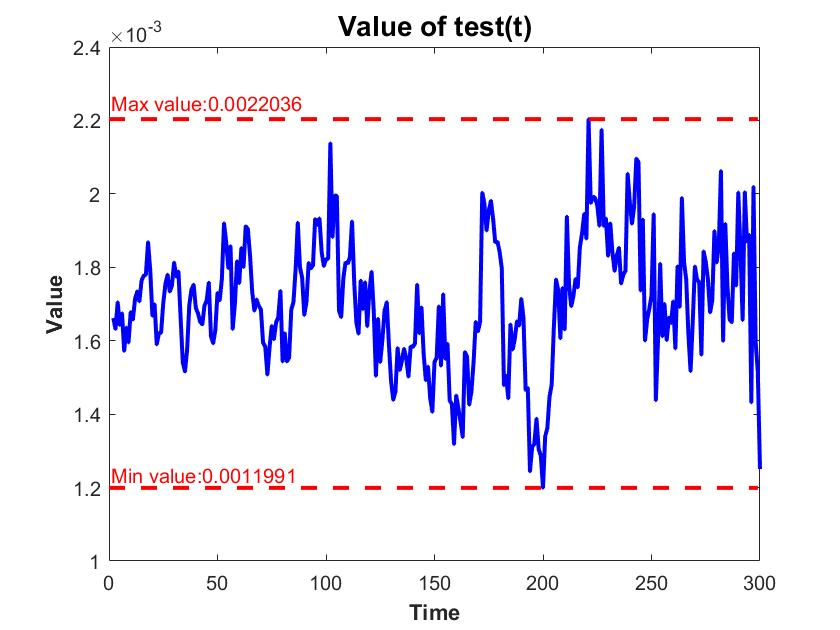}
        \caption{Experiment 4}
        \label{fig:image4}
    \end{subfigure}
    \caption{Incompressible test for Experiment 3 (left) and Experiment 4 (right)}
    \label{fig:both2}
\end{figure}

For more complex three-dimensional cases, the incompressibility assumption remains valid. Moreover, fluid systems with higher velocities tend to exhibit larger values of $\text{test}(t)$ than systems with lower velocities, consistent with the trends observed in the two-dimensional cases.

\subsubsection{Comparison with other methods}\label{comparison}
The method proposed in this study is relatively novel and offers considerable scope for further optimization. Consequently, in this subsection, we compare and analyze our method with other available methods to identify the advantages of our approach. However, the primary objective of this work is to introduce the new approach and establish its validity, rather than to compare it against fully optimized and well-established methods, so we compare our method with early versions of other methods. To provide a preliminary benchmark, two classical numerical techniques are employed: the finite difference method (FDM) and the finite volume method (FVM), both of which are widely adopted and critically important in computational fluid dynamics. The FDM approximates the derivative terms in differential equations using finite differences, constructing discrete equations from function values at grid nodes. In contrast, the FVM divides the computational domain into a set of control volumes and derives discrete equations by integrating the governing differential equations over each control volume. What's more, we also compare Random LES method with mild solution method mentioned in Section 1, which is widely studied in the theory of stochastic PDE. Using Experiment 4 as an illustrative example, we present an analysis comparing the performance of these methods. All simulations are conducted on Intel(R) i7-12700H CPU (2.30 GHz) .

For FDM method, we obey the assumption that $\Delta t \leq \frac{\Delta x}{|U_max|}$, which means that for small mesh size, we should select a small timestep $\Delta t$ to enable the stability of numerical method. For the FVM method, we use the package in python to finish the simulation. Then we add the simulation utilizing the mild solution mentioned in Section 1, the formula is:
\begin{align*}
    &u^i(x,t)-u^i_0(x)= \\
    &\int_0^t \int_D\left(u^i(y,s) u(y,s)\cdot \nabla_y \ln{h}(x,t-s,y)-{\partial\over\partial y_i} p(y,s) - F^i(y,s)\right)h(x,t-s,y)\textrm{d}y\textrm{d}s,
    \end{align*}
for $i=1,2,3$.
The velocity $u(x,t)$ is updated by the variables in time $t-1$, the calculation of $\nabla p$ is calculated by the same method (kernel method) in Random LES and the gradient of $u$ needs to be calculated by FDM method since we can not gain the explicit expression of it. We detect the stability of all the numerical methods by observing when the numerical solution blows up. The results are shown in Table \ref{test1}.

\begin{table}[H]

\centering  
\caption{Numerical instability test}
\resizebox{\textwidth}{!}{
\begin{tabular}{|Sc|Sc|Sc|Sc|Sc|}
  \hline
  &Random LES method& FDM method& FVM method & Mild solution\\
  \hline
   \makecell{$\Delta x=\Delta y=\frac{3\pi}{50}$, \\ $\Delta z=\frac{0.2\pi}{50}$ , $\Delta t=0.01$}&Stable & Explode at step 10 & Stable & Explode at step 22\\
   \hline
   \makecell{$\Delta x=\Delta y=\frac{3\pi}{50}$, \\$\Delta z=\frac{0.2\pi}{50}$ , $\Delta t=0.001$} &Stable & Explode at step 21 & Stable & Explode at step 45\\
   \hline
   \makecell{$\Delta x=\Delta y=\frac{3\pi}{100}$, \\$\Delta z=\frac{0.2\pi}{100}$ , $\Delta t=0.001$} &Stable & Explode at step 41 & Stable & Explode at step 75\\
  \hline
\end{tabular}}\label{test1}
\end{table}

It is evident that, without the incorporation of complex boundary constraints, both the FDM and mild solution methods are prone to numerical instability and potential blow-up. The latter tends to diverge more slowly, as it involves only first-order derivative approximations and avoids the computation of second-order derivatives. In contrast, the Random LES method and FVM demonstrate greater stability. While numerous researchers continue to explore various optimizations for the FDM, our analysis focuses on comparisons with its baseline versions. Theoretically, our Monte Carlo-based approach is less susceptible to numerical instability. Subsequently, we evaluated the average runtime per epoch, with the results presented in Table \ref{test2}.

\begin{table}[H]
\centering  
\caption{Average calculation time of each epoch}

\begin{tabular}{|Sc|Sc|Sc|Sc|Sc|}
  \hline
  &Random LES method& FDM method& FVM method &Mild solution\\
  \hline
   \makecell{$\Delta x=\Delta y=\frac{3\pi}{50}$, \\ $\Delta z=\frac{0.2\pi}{50}$, $\Delta t=0.01$}  & 4.69 seconds & 0.449 seconds & 6.5 seconds &26.5 seconds\\
   \hline
   \makecell{$\Delta x=\Delta y=\frac{3\pi}{50}$, \\$\Delta z=\frac{0.2\pi}{50}$ , $\Delta t=0.001$} &7.71 seconds & 1.143 seconds & 13 seconds &73.4 seconds\\
   \hline
   \makecell{$\Delta x=\Delta y=\frac{3\pi}{100}$, \\$\Delta z=\frac{0.2\pi}{100}$ , $\Delta t=0.001$} &16.5 seconds & 11.54 seconds & 78.7 seconds &165.4 seconds\\
  \hline
\end{tabular}\label{test2}
\end{table}

In terms of computational time, the FDM, benefiting from its integration into existing software packages, generally achieves faster computations compared to our approach. Nonetheless, we argue that the computational cost of our method remains well within acceptable limits. More importantly, as grid resolution increases, FDM's computational time escalates significantly and is accompanied by an increased risk of numerical instability. In contrast, our method exhibits stable computational scalability without these limitations. The mild solution method requires the longest computational time among the compared methods, with runtime increasing substantially as mesh resolution becomes finer. Because the kernel term and the other terms are time-reverse, which indicates that we need to recalculate integral from 0 to $t$ in every iterative step in mild solution method, which is very expensive in calculation. Finally, we present a theoretical analysis of the computational complexity of the proposed numerical scheme.

\begin{table}[H]
\centering  
\caption{Analysis of complexity}
\resizebox{\textwidth}{!}{
\begin{tabular}{|c|c|c|c|c|}
  \hline
  Updating variable &Random LES method& FDM method& FVM method & Mild solution\\
  \hline
   Brownian particle $Y_t$ & Using value at $t-1$  & \textbackslash & \textbackslash & \textbackslash\\
   \hline
   Velocity $U_t$ &Using 1-order derivative & Using 2-order derivative & Using 2-order derivative & Using 1-order derivative\\
   \hline
   Pressure $P$ & \textbackslash & Using 2-order derivative& Using 2-order derivative & \textbackslash\\
   \hline
   Gradient of $P$ & Using 1-order derivative & \textbackslash & \textbackslash & Using 1-order derivative\\
  \hline
\end{tabular}}\label{test3}
\end{table}

We observe that the numerical scheme of the Random LES method requires only the computation of first-order derivatives, while both the FDM and FVM involve second-order derivatives. Theoretically, the Random LES approach offers distinct advantages: calculating second-order derivatives not only leads to greater error accumulation but also increases the risk of numerical instability and entails higher computational cost, particularly in high-dimensional settings. Although the mild solution method relies solely on first-order derivatives, the term $\nabla u$ cannot be explicitly expressed. Therefore, we still need to employ additional numerical techniques to approximate it.

According to the results in Table \ref{test1}, \ref{test2}, \ref{test3} and the corresponding analysis, compared to alternative approaches, the Random LES method generally demonstrates superior stability, acceptable computational efficiency, and relatively lower theoretical complexity. A key advantage of our approach lies in its ability to reveal diverse potential structures of the fluid system through individual simulations. Since the exact morphology of the fluid cannot be perfectly determined through computation, conventional deterministic methods that produce fixed solutions may fail to capture significant fluid behaviors emerging from our probabilistic framework. Our Random LES method, which is based on explicit variable representation and probabilistic format, significantly improves the accuracy and stability of numerical experiments without substantially increasing computational time.

\newpage

\appendix
\renewcommand{\thesection}{\Alph{section}}
\newcommand{\sectionname}{Section}
\renewcommand{\sectionname}{Appendix}

\makeatletter
\renewcommand{\@seccntformat}[1]{\sectionname~\thesection: }
\makeatother
\section{Fluid flows represented by Brownian fluid particles}\label{representation}
In this section we establish the technical facts about divergence-free
vector fields, and their Taylor's diffusion, which are used in the
paper.

The following convention will be applied throughout the paper unless
otherwise specified. Suppose $b(x,t)$ be a time-dependent vector
field on $\mathbb{R}_{+}^{d}$ such that $\nabla\cdot b=0$ and $b(x,t)=0$
for $x=(x_{1},\cdots,x_{d})$ with $x_{d}=0$, then $b(x,t)$ is extended
to the whole space $\mathbb{R}^{d}$ by reflection: $b^{i}(\bar{x},t)=b^{i}(x,t)$
for $i=1,\cdots,d-1$ and $b^{d}(\bar{x},t)=-b^{d}(x,t)$ for any
$x\in\mathbb{R}^{d}$ and $t\geq0$, where $x\mapsto\bar{x}=(x_{1},\cdots,x_{d-1},-x_{d})$
is the reflection about $x_{d}=0$. The divergence-free feature of
$u$ is retained though in weak sense, that is, $\nabla\cdot b=0$
in distribution on $\mathbb{R}^{d}$, which implies that the $L^{2}$-adjoint
of the (forward) heat operator $L_{b}-\frac{\partial}{\partial t}$
coincides with the (backward) heat operator $L_{-b}+\frac{\partial}{\partial t}$.
Here, given a time-dependent vector field $b(x,t)$, $L_{b}=\nu\Delta+b\cdot\nabla$
denotes the time-dependent elliptic operator of second-order in $\mathbb{R}^{d}$.
Following an idea of Taylor \citep{Taylor1921} that determining a
vector field $b(x,t)$ is equivalent to the description of `imaginary'
fluid particles with the velocity $b(x,t)$, we may consider the diffusions
defined by It\^o's stochastic differential equation 
\begin{equation}
\textrm{d}X=b(X,t)\textrm{d}t+\sqrt{2\nu}\textrm{d}B,\label{b-SDE-1}
\end{equation}
where $B$ is a Brownian motion on some probability space. The weak
solution (cf. \citep{StroockVaradhan} and \citep{Ikeda1989}) of
(\ref{b-SDE-1}) defines the diffusion with the infinitesimal generator
$L_{b}$, cf. \citep{Friedman1964}. Let $p_{b}(s,x;t,y)$ and $p_{b}^{+}(s,x;t,y)$
(for $t>s\geq0$) be the transition probability density function of
the $L_{b}$-diffusion and, respectively, of the $L_{b}$-diffusion
stopped on leaving the region $\mathbb{R}_{+}^{d}$. Formally $p_{b}(s,x;t,y)$
is the probability that the diffusion $X_{t}$ hits $y$ given that
$X_{s}=x$, and similarly $p_{b}^{+}(s,x;t,y)$ the probability that
$X_{t}$ hits $y$ before leaving the domain $\mathbb{R}_{+}^{d}$
given that $X_{s}=x$. Then it holds that 
\begin{equation}
p_{b}^{+}(s,x;t,y)=p_{b}(s,x;t,y)-p_{b}(s,x;t,\bar{y})\quad\forall x,y\in\mathbb{R}_{+}^{d}.\label{pD-p}
\end{equation}
Given $\eta\in\mathbb{R}^{d}$, the distribution of $X$, a weak solution
to (\ref{b-SDE-1}) that $X_{0}=\eta$, is denoted by $\mathbb{P}^{\eta}$,
which is a probability measure on the space $C$({[}0,$\infty),\mathbb{R}^{d})$
of continuous paths. Consider for every $\xi$ , It\^o's stochastic
differential equation 
\begin{equation}
\textrm{d}X_{t}=b(X_{t},t)\textrm{d}t+\sqrt{2\nu}\textrm{d}B_{t},\quad X_{0}=\xi,\label{b-SDE}
\end{equation}
where $B_{t}$ is $d$-dimensional Brownian motion. In addition to
It\^o's SDE formulation, we shall use the weak solution formulation
as well, cf. \citep{StroockVaradhan}. Let $\varOmega=C([0,\infty),\mathbb{R}^{d})$
be the continuous path space in $\mathbb{R}^{d}$, and $X=(X_{t})_{t\geq0}$
denote the coordinate process, that is, for each $t\geq0$, $X_{t}:\varOmega\mapsto\mathbb{R}^{d}$,
which sends each path $\psi\in\varOmega$ to $\psi(t)$. For simplicity,
$X_{t}$ shall be written as $\psi(t)$ if no confusion may arise
from the context. Let $\mathcal{F}_{t}^{0}=\sigma\{X_{s}:s\leq t\}$
be the smallest $\sigma$-algebra on $\varOmega$ such that $X_{s}$
are measurable for all $s\leq t$, and $\mathcal{F}^{0}=\sigma\{X_{s}:s<\infty\}$.
Then $\mathcal{F}^{0}=\mathcal{B}(\varOmega)$ the Borel $\sigma$-algebra
on $\varOmega$ generated by the uniform convergence over any bounded
subset of $[0,\infty)$.

For simplicity let us assume that $b(x,t)$ is bounded and jointly
Borel measurable. According to \citep{StroockVaradhan}, for each
$\xi\in\mathbb{R}^{d}$ and each $\tau\geq0$, there is a unique probability
measure $\mathbb{P}_{b}^{\xi,\tau}$ on $(\varOmega,\mathcal{F}^{0})$
such that $\mathbb{P}_{b}^{\xi,\tau}\left[X_{s}=\xi\textrm{ for }s\leq\tau\right]=1$,
and 
\[
M_{t}^{[f]}=f(X_{t},t)-f(X_{\tau},\tau)-\int_{\tau}^{t}(L_{b}f)(X_{s},s)\textrm{d}s,
\]
is a martingale (for $t\geq\tau$ and under the probability $\mathbb{P}_{b}^{\xi,\tau}$)
for every $f\in C_{b}^{2,1}(\mathbb{R}^{d}\times[0,\infty))$. For
simplicity $\mathbb{P}_{b}^{\xi,0}$ is denoted by $\mathbb{P}_{b}^{\xi}$.
The probability measure $\mathbb{P}_{b}^{\xi,\tau}$ is called the
$L_{b}$-diffusion, or Taylor's diffusion, started from $\xi$ at
instance $\tau$. $L_{b}$ is also called the infinitesimal generator
of the diffusion $\mathbb{P}_{b}^{\xi,\tau}$. Since $L_{b}$ is uniformly
elliptic, the transition probability $P_{b}(\tau,\xi;t,\textrm{d}x)$
of the $L_{b}$-diffusion, which is the distribution of $X_{t}$ under
$\mathbb{P}_{b}^{\xi,\tau}$, where $t>\tau\geq0$, has a positive
and H\"older's continuous density denoted by $p_{b}(\tau,\xi;t,x)$
(for $t>\tau\geq0$, $\xi,x\in\mathbb{R}^{d}$). That is, $p_{b}(\tau,\xi;t,x)\textrm{d}x=\mathbb{P}_{b}^{\xi,\tau}\left[X_{t}\in\textrm{d}x\right]$.

Let $T>0$ be given. The conditional distribution $\mathbb{P}_{b}^{\xi}\left[\left.\cdot\right|X_{T}=\eta\right]$
can be constructed as the following. Let 
\[
q_{b}(s,x;t,y)=\frac{p_{b}(s,x;t,y)p_{b}(t,y;T,\eta)}{p_{b}(s,x;T,\eta)},
\]
for $0\leq s<t<T$ and $x,y\in\mathbb{R}^{d}$, which is a transition
probability density function. The diffusion associated with the transition
probability density function $q_{b}$ started from $\xi$ at $\tau=0$
shall be denoted by $\mathbb{P}^{\xi,0\rightarrow\eta,T}$, which
shall be consider as the probability measure $\varOmega=C([0,T],\mathbb{R}^{d})$.

Let $b^{T}(x,t)=b(x,(T-t)^{+})$ and $\tau_{T}$ be the time reversal
operation on $\varOmega$, that is, $\psi\circ\tau_{T}(t)=\psi(T-t)$
for $t\in[0,T]$. It can be verified that 
\begin{equation}
p_{-b^{T}}(s,x;t,y)=p_{b}(T-t,y;T-s,x),\label{div-b-01}
\end{equation}
for $0\leq s<t\leq T$, and $x,y\in\mathbb{R}^{d}$. The duality of
conditional diffusion laws has been established in \citep{Qian2022}. 
\begin{lem}
\label{lem31}Suppose $b(x,t)$ is divergence-free, that is, $\nabla\cdot b=0$
(for every $t$) in the sense of distribution on $\mathbb{R}^{d}$.
Then 
\[
\mathbb{P}_{b}^{\xi,0\rightarrow\eta,T}=\mathbb{P}_{-b^{T}}^{\eta,0\rightarrow\xi,T}\circ\tau_{T}
\]
for every $\xi,\eta\in\mathbb{R}^{d}$. 
\end{lem}

From now on, we assume that $b(x,t)$ for $x\in\mathbb{R}_{+}^{d}$
and $t\geq0$ be a time-dependent vector field on $D$, satisfying
that $b(x,t)=0$ for $x\in\partial\mathbb{R}_{+}^{d}$ (the non-slip
condition). Suppose that $b(x,t)$ is differentiable up to the boundary
$\partial\mathbb{R}_{+}^{d}$ and is bounded. For each $t$, the vector
field $b(x,t)$ is extended to the whole space $\mathbb{R}^{d}$ via
the reflection about the hyperspace $x_{d}=0$, so that 
\begin{equation}
b^{i}(x,t)=b^{i}(\bar{x},t)\quad\textrm{ for }i=1,\ldots,d-1\textrm{ and }b^{d}(x,t)=-b^{d}(\bar{x},t).\label{v1-01}
\end{equation}
Then $\overline{b(x,t)}=b(\bar{x},t)$ for all $x\in\mathbb{R}^{d}$
and $t\geq0$. We also assume that $b(x,t)$ is divergence-free on
$\mathbb{R}_{+}^{d}$. Then $\nabla\cdot b=0$ in the sense of distribution
on $\mathbb{R}^{d}$. 
\begin{thm}
\label{thm23-5} Let $\varPsi$ be a solution to the parabolic equation
\begin{equation}
\left(L_{-b}-\frac{\partial}{\partial t}\right)\varPsi+g=0\quad\textrm{ in }\mathbb{R}_{+}^{d},\label{par-01-1}
\end{equation}
subject to the non-slip condition that $\varPsi(x,t)=0$ for $x=(x_{1},\cdots,x_{d})$
with $x_{d}=0$. Then
\begin{align}
\varPsi(\xi,t) & =\int_{\mathbb{R}_{+}^{d}}\left(p_{b}(0,\eta;t,\xi)-p_{b}(0,\bar{\eta};t,\xi)\right)\varPsi_{0}(\eta)\textrm{d}\eta\nonumber \\
 & +\int_{0}^{t}\int_{\mathbb{R}_{+}^{d}}\mathbb{P}_{b}^{\eta,0\rightarrow\xi,t}\left[1_{\{s>\gamma_{t}(X^{\eta})\}}g(\psi(s),s)\right]p_{b}(0,\eta;t,\xi)\textrm{d}\eta\textrm{d}s,\label{W-aa2-1-1}
\end{align}
for $t>0$ and $\xi\in\mathbb{R}_{+}^{d}$, where $\varPsi_{0}=\varPsi(\cdot,0)$,
$\zeta(\psi)=\inf\{s:\psi(s)\notin\mathbb{R}_{+}^{d}\}$ and $\gamma_{t}(\psi)$
denotes $\sup\{s\in(0,t):\psi(s)\in\mathbb{R}_{+}^{d}\}$ respectively. 
\end{thm}

\begin{proof}
Let $D=\mathbb{R}_{+}^{d}$. The proof relies on the duality Lemma
\ref{lem31}, established in \citep{Qian2022}. Let $T>0$ and
$\xi\in D$ be any but fixed. Let $b^{T}(x,t)=b(x,(T-t)^{+})$, and
$Y$ denote the weak solution to SDE 
\begin{equation}
\textrm{d}Y_{t}=-b(Y_{t},T-t)\textrm{d}t+\sqrt{2\nu}\textrm{d}B_{t},\quad Y_{0}=\xi,\label{ba-s01}
\end{equation}
whose infinitesimal generator is $L_{-b^{T}}$. Then 
\begin{equation}
Y_{t}=\xi-\int_{0}^{t\wedge T}b(Y_{s},T-s)\textrm{d}s+\sqrt{2\nu}\int_{0}^{t\wedge T}\textrm{d}B_{s}\quad\textrm{ for all }t\geq0.\label{sint-01}
\end{equation}
Let $T_{\xi}=\inf\left\{ t\geq0:Y_{t}\notin D\right\} $. Then 
\begin{equation}
Y_{t\wedge T_{\xi}}=\xi-\int_{0}^{t}1_{\left\{ s<T\wedge T_{\xi}\right\} }b(Y_{s},T-s)\textrm{d}s+\sqrt{2\nu}\int_{0}^{t}1_{\left\{ s<T\wedge T_{\xi}\right\} }\textrm{d}B_{s},\label{Y-b}
\end{equation}
for all $t\geq0$. Let 
\[
Z_{t}=\varPsi(Y_{t\wedge T_{\xi}},T-t)=1_{\{t<T_{\xi}\}}\varPsi(Y_{t},T-t)\quad\textrm{ for }t\leq T,
\]
where the second equality follows from the non-slip condition: $\varPsi$
vanishes along the boundary $\partial D$. According to Itô's formula,
\[
\textrm{d}Z_{t}=\nabla\varPsi(Y_{t\wedge T_{\xi}},T-t)\cdot\textrm{d}Y_{t\wedge T_{\xi}}-\frac{\partial\varPsi}{\partial t}(Y_{t\wedge T_{\xi}},T-t)\textrm{d}t+\nu1_{\left\{ t<T\wedge T_{\xi}\right\} }\Delta\varPsi(Y_{t},T-t)\textrm{d}t.
\]
Using the non-slip condition again: $\varPsi$ vanishes identically
on $\partial D$, so that 
\[
\frac{\partial\varPsi}{\partial t}(Y_{t\wedge T_{\xi}},T-t)=1_{\left\{ t<T_{\xi}\right\} }\frac{\partial\varPsi}{\partial t}(Y_{t},T-t)\quad\textrm{ for }t\leq T.
\]
Therefore 
\begin{align}
\textrm{d}Z_{t} & =\sqrt{2\nu}1_{\left\{ t<T_{\xi}\right\} }\nabla\varPsi(Y_{t},T-t)\cdot\textrm{d}B_{t}\nonumber \\
 & +1_{\left\{ t<T_{\xi}\right\} }\left(\nu\Delta-b\cdot\nabla-\frac{\partial}{\partial t}\right)\varPsi(Y_{t},T-t)\textrm{d}t,\label{Yj-06}
\end{align}
for $t\leq T$. Since $\varPsi$ solves the parabolic equations, so
that 
\begin{align}
Z_{T\wedge T_{\xi}} & =\varPsi(\xi,T)-\int_{0}^{T\wedge T_{\xi}}1_{\{t<T_{\xi}\}}g(Y_{t},T-t)\textrm{d}t\nonumber \\
 & +\sqrt{2\nu}\int_{0}^{T\wedge T_{\xi}}1_{\{t<T_{\xi}\}}\nabla\varPsi(Y_{t},T-t)\cdot\textrm{d}B.\label{ba-0051}
\end{align}
Taking expectation both sides we obtain 
\begin{equation}
\varPsi(\xi,T)=\mathbb{E}\left[\varPsi(Y_{T},0)1_{\{T<T_{\xi}\}}\right]+\mathbb{E}\left[\int_{0}^{T}1_{\{t<T_{\xi}\}}g(Y_{t},T-t)\textrm{d}t\right].\label{Wi-004}
\end{equation}
The first term $J_{1}=\mathbb{E}\left[\varPsi_{0}(Y_{T})1_{\{T<T_{\xi}\}}\right]$
on the right hand side may be written as 
\begin{align*}
J_{1} & =\int_{D}\varPsi_{0}(\eta)p_{-b^{T}}^{D}(0,\xi;T,\eta)\textrm{d}\eta\\
 & =\int_{D}\varPsi_{0}(\eta)\left(p_{-b^{T}}(0,\xi;T,\eta)-p_{-b^{T}}(0,\xi;T,\bar{\eta})\right)\textrm{d}\eta\\
 & =\int_{D}\left(p_{b}(0,\eta;T,\xi)-p_{b}(0,\bar{\eta};T,\xi)\right)\varPsi_{0}(\eta)\textrm{d}\eta.
\end{align*}
The second term $J_{2}=\mathbb{E}\left[\int_{0}^{T}g(Y_{t},T-t)1_{\{t<T_{\xi}\}}\textrm{d}t\right]$
on the right hand side can be treated similarly 
\begin{align*}
J_{2} & =\int_{0}^{T}\int_{\mathbb{R}^{3}}\mathbb{E}\left[\left.1_{\{t<T_{\xi}\}}g(Y_{t},T-t)\right|Y_{T}=\eta\right]p_{-b_{T}}(0,\xi;T,\eta)\textrm{d}t\\
 & =\int_{0}^{T}\int_{\mathbb{R}^{3}}\mathbb{P}_{-b^{T}}^{\xi,0\rightarrow\eta,T}\left[1_{\{t<\zeta(\psi)\}}g(\psi(t),T-t)\right]p_{b}(0,\eta;T,\xi)\textrm{d}\eta\textrm{d}t\\
 & =\int_{0}^{T}\int_{\mathbb{R}^{3}}\mathbb{P}_{b}^{\eta,0\rightarrow\xi,T}\left[1_{\{t<\zeta(\psi\circ\tau_{T})\}}g(\psi(T-t),T-t)\right]p_{b}(0,\eta;T,\xi)\textrm{d}\eta\textrm{d}t\\
 & =\int_{0}^{T}\int_{\mathbb{R}^{3}}\mathbb{P}_{b}^{\eta,0\rightarrow\xi,T}\left[1_{\{T-t<\zeta(\psi\circ\tau_{T})\}}g(\psi(t),t)\right]p_{b}(0,\eta;T,\xi)\textrm{d}\eta\textrm{d}t,
\end{align*}
where we have used the fact that, since $\nabla\cdot b=0$, $p_{-b_{T}}(0,\xi,T,\eta)$
coincides with $p_{b}(0,\eta,T,\xi)$. Therefore 
\begin{align}
\varPsi(\xi,T) & =\int_{D}\left(p_{b}(0,\eta;T,\xi)-p_{b}(0,\bar{\eta};T,\xi)\right)\varPsi_{0}(\eta)\textrm{d}\eta.\nonumber \\
 & +\int_{0}^{T}\int_{\mathbb{R}^{3}}\mathbb{P}_{b}^{\eta,0\rightarrow\xi,T}\left[1_{\{T-t<\zeta(\psi\circ\tau_{T})\}}g(\psi(t),t)\right]p_{b}(0,\eta;T,\xi)\textrm{d}\eta\textrm{d}t.\label{WW-002}
\end{align}
\end{proof}
By using a similar but much simpler proof, we also have the following
representation.
\begin{thm}
Let $\varPsi$ be a solution to the parabolic equation
\begin{equation}
\left(L_{-u}-\frac{\partial}{\partial t}\right)\varPsi+g=0\quad\textrm{ in }\mathbb{R}^{d}.\label{par-01-1-1}
\end{equation}
Then 
\begin{align}
\varPsi(\xi,t) & =\int_{\mathbb{R}^{d}}p_{b}(0,\eta;t,\xi)\varPsi(\eta,0)\mathrm{d}\eta\nonumber \\
 & +\int_{0}^{t}\int_{\mathbb{R}^{d}}\mathbb{P}_{b}^{\eta,0\rightarrow\xi,t}\left[f(X_{s},s)\right]p_{b}(0,\eta;t,\xi)\mathrm{d}\eta\mathrm{d}s,
\end{align}
for every $\xi\in\mathbb{R}^{d}$.
\end{thm}

\section*{Data Availability Statement}

The data that support the findings of this study are available from
the corresponding author upon reasonable request.

\section*{Declaration of Interests}

The authors report no conflict of interest.



\end{document}